\documentclass[11pt]{article}
\usepackage{epsfig}
\usepackage{amsmath}
\usepackage{xspace}
\usepackage{amssymb}
\usepackage{graphicx}
\usepackage{latexsym}
\usepackage{subfigure}
\usepackage{verbatim}
\usepackage{hyperref} 
\usepackage[all]{hypcap} 
\usepackage{caption} 
\usepackage{algorithm}
\usepackage{algorithmic}
\usepackage{multirow}
\usepackage{pbox}
\usepackage{pdfpages}
\usepackage[toc,page]{appendix}
\usepackage{comment}
\usepackage[english]{babel}

\usepackage{listings}
\usepackage[OT4,T1]{fontenc}
\usepackage{url}
\usepackage{diagbox}

\newcounter{alg}
\newfloat{algorithm}{H}{loa}
\floatname{algorithm}{Algorithm}
\newtheorem{remark}{Remark}[section]
\newtheorem{example}[remark]{Example}

\newcommand\Eq{\mathcal{E}}

\newcommand\G{\mathcal{G}}

\allowdisplaybreaks

\newtheorem{theorem}{Theorem}[section]

\newtheorem{claim}[theorem]{Claim}

\newtheorem{definition}{Definition}[section]

                      {}

                      {}

\def\squarebox#1{\hbox to #1{\hfill\vbox to #1{\vfill}}}
\newcommand{\qed}{\hspace*{\fill}
\vbox{\hrule\hbox{\vrule\squarebox{.667em}\vrule}\hrule}\smallskip}

\def\eod{\vrule height 6pt width 5pt depth 0pt}
\newenvironment{proof}{\noindent {\bf Proof:} \hspace{.677em}}
                      {\hspace*{\fill}{\eod}}

\DeclareMathOperator*{\argmin}{arg\min}

\newcommand{\ceil}[1]{\left\lceil #1 \right\rceil}

\newcommand{\floor}[1]{\left\lfloor #1 \right\rfloor}

\begin{document}

\thispagestyle{empty}
\title{Coordination Mechanisms with Rank-Based Utilities
}
\author{Gilad Lavie$^{*}$ \and Tami Tamir\thanks{School of Computer Science, Reichman University, Israel. E-mail: giladlavie@gmail.com, tami@runi.ac.il.}}
\maketitle

 \begin{abstract}
In classical job-scheduling games, each job behaves as a selfish player, choosing a machine to minimize its own completion time. To reduce the equilibria inefficiency, coordination mechanisms are employed, allowing each machine to follow its own scheduling policy. In this
paper we study the effects of incorporating {\em rank-based utilities} within coordination mechanisms across environments with either identical or unrelated machines.
With rank-based utilities, players aim to perform well relative to their competitors, rather than solely minimizing their completion time.

We first demonstrate that even in basic setups, such as two identical machines with unit-length jobs, a pure Nash equilibrium (NE) assignment may not exist. This observation motivates our inquiry into the complexity of determining whether a given game instance admits a NE. We prove that
this problem is NP-complete, even in highly restricted cases. In contrast, we identify specific classes of games where a NE is guaranteed to exist, or where the decision problem can be resolved in polynomial time.
Additionally, we examine how competition impacts the efficiency of Nash equilibria, or sink equilibria if a NE does not exist. We derive tight bounds on the price of anarchy, and show that competition may either enhance or degrade overall performance.

\end{abstract}

\section{Introduction}
Scheduling problems have traditionally been studied from a centralized point of view in which the goal is to find an assignment of jobs to machines so as to minimize some global objective function. In practice, many resource allocation services lack a central authority, and are often managed by multiple strategic users, whose individual payoff is affected by the assignment of other users. As a result, game
theory has become an essential tool in the analysis of scheduling environments. This stems from the understanding that customers (the jobs) as well as resource owners (the machines) act independently so as to maximize their own benefit.

Job-scheduling games are singleton congestion games that represent situations which commonly occur in roads, and communication networks. In these well-studied models, each job acts as a selfish player, choosing a machine to minimize its own completion time. An algorithmic tool that is commonly utilized by the designer of such a system is a {\em coordination mechanism} \cite{CKN04}. The coordination mechanism uses a scheduling policy within each machine that aims to mitigate the impact of selfishness to performance. 
The scheduling policy defines the order according to which jobs are scheduled within a machine. For example, the jobs assigned to a machine that applies a {\em ShortestFirst} policy, are processed from shortest to longest. The assignment is non-preemptive and each job is processed uninterruptedly.

A coordination mechanism induces a game in which the strategy space of each player is the set of machines. An {\em assignment} is a strategy profile, defined by the jobs' selections. Every machine processes the jobs assigned to it according to its scheduling policy. The profile induces a completion time for each job - the time when its processing is done.

Traditionally, in such scheduling games, the goal of a player is to select a machine such that its completion time is minimized (Christodoulou et al. \cite{CKN04}, Immorlica et al. \cite{ILMS09}). In this work, inspired by Rosner and Tamir~\cite{RT23}, we study {\em coordination mechanisms with rank-based utilities}. Formally, all players are competitors. The player's main goal is to do well relative to its competitors, i.e., to minimize the {\em rank} of its completion time among all players, while minimizing the completion time itself is a secondary objective. This natural objective arises in several computational environments. For example, $(i)$ in high-frequency trading, the absolute time to execute a trade is irrelevant as long as it is faster than competitors, since only the first trade secures a profit. The focus is entirely on being ahead of others. The competitive nature of this environment ensures that speed is the sole priority, as market conditions can shift so rapidly that even minor delays result in missed opportunities.

$(ii)$ In cloud computing, users compete for access to shared servers, where the allocation of computational resources is often determined by relative performance metrics. For instance, during high-demand events like Black Friday, where e-commerce platforms experience surges in traffic, users with higher priority - such as premium subscribers or those leveraging performance-optimized instances - receive faster task processing. In these scenarios, the focus is not on minimizing absolute completion time but on being processed ahead of others to gain a competitive edge. This prioritization is especially critical in time-sensitive operations, such as real-time analytics or transaction processing, where delays can significantly impact business outcomes.

$(iii)$ On e-commerce platforms like Amazon, sellers optimize metrics such as response and delivery times to achieve higher rankings in search results. These rankings are based on relative performance, where sellers with faster responses and better fulfillment times gain increased visibility. Higher visibility directly translates into increased sales, creating a competitive environment where outperforming others matters more than achieving an absolute delivery speed.


An assignment is a \emph{pure Nash equilibrium} (NE) if no player can benefit from unilaterally deviating from its strategy, that is, no player has an incentive to deviate from the machine on which its job is assigned. In games with rank-based utilities, a NE is a profile in which no job can reduce the rank of its completion time, nor to keep its rank and reduce its completion time.

In this work we study several aspects of coordination mechanisms with rank-based utilities. Our results focus on the challenge of determining whether a pure NE exists. We prove that this problem is NP-complete, even for very restricted classes of games, where job-lengths are limited to be in $\{1,2\}$. On the other hand, we identify several classes of games where the problem can be solved in polynomial time, and we characterize specific game classes where a NE is guaranteed to exist. For games with a global priority list, we provide a clear characterization of instances that have a NE. For games that have no NE, we propose a dynamic of best-
response deviations with a fast convergence to an efficient sink equilibrium. Additionally, we show that for games with unit-length jobs, two machines, and machine-dependent priority lists, deciding the existence of a NE and constructing one can be done in linear time. We also introduce the {\em Inversed-Policies} coordination mechanism, which ensures the existence of a NE in games with two identical machines.

Our analysis also examines how competition impacts the efficiency of equilibria, with respect to the makespan. We extend classical efficiency bounds for traditional scheduling games to games with rank-based utilities, and show that competition can either improve or worsen performance. Finally, we extend our analysis to games where players are divided into competition sets, demonstrating that certain partitions can ensure or prevent the existence of a NE, and in modified models with competition-class-based priorities, a NE is always guaranteed to exist, with Best Response Dynamics (BRD) converging to equilibrium.



\subsection{Notation and Problem Statements}
An instance of a {\em coordination mechanism with rank-based utilities} is given by a tuple $G=\langle J,M,(p_j)_{j \in J}, (r_i)_{M_i \in M},(\pi_i)_{M_i \in M} \rangle$, where $J$ is a finite set of $n\geq 1$ jobs (also denoted {\em players}), $M$ is a finite set of $m\geq 1$ machines, $p_j\in\mathbb{R}^+$ is the processing time of job $j\in J$, $r_i\in\mathbb{R}^+$ denotes the rate or speed of machine $M_i \in M$, and $\pi_i:J \rightarrow \{1,\ldots,n\}$ is the {\em priority list} of machine $M_i\in M$. We assume that jobs have no release times. We denote by $j_1 \prec_i j_2$ the fact that job $j_1$ is prioritized over job $j_2$ in the priority list of machine $M_i$, that is $\pi_i(j_1)< \pi_i(j_2)$. For every job $j \in J$, the other jobs in $J$ are referred to as the competitors of $j$.

A strategy profile $s=(s_j)_{j\in J}\in M^{|J|}$ assigns a machine $s_j\in M$ to every job $j\in J$. A strategy profile $s$ induces a {\em schedule} in which the jobs are processed according to their order in the machines' priority lists. We use $s$ to denote both the strategy profile and its induced schedule.
The set of jobs that delay $j\in J$ in $s$ is denoted by $E_{j}(s)= \{j' \in J| s_{j'}=s_j \wedge \pi_{s_j}(j') \leq \pi_{s_j}(j)\}$. Note that job $j$ itself also belongs to $E_{j}(s)$. Let $P_j(s)=\sum\limits_{j'\in E_{j}(s)}p_{j'}$. The completion time of job $j\in J$ is given by $C_j(s)=P_j(s)/{r_{s_j}}$. For machine $M_i \in M$ let $L_i(s)= \sum_{j|s_j=i} p_j$ be the load on $M_i$ in $s$, i.e., the total length of the jobs assigned to it.

Unlike classical job-scheduling games, in which the goal of a player is to minimize its completion time, in games with rank-based utilities, the goal of a player is to do well relative to its competitors. That is, every profile induces a ranking of the players according to their completion time, and the goal of each player is to have the lowest possible rank among all players. Formally, for a profile $s$, let $C^s = \langle C_1^s, \ldots, C_n^s \rangle$ be a sorted vector of the completion times of the players in $J$. That is, $C_1^s \leq \ldots \leq C_n^s$, where $C_1^s$ is the minimal completion time of a player in $s$, etc. The $rank$ of player $j \in J$ in profile $s$, denoted by $rank_j(s)$, is the rank of its completion time in $C^s$. If several players have the same completion time, then they all have the same rank, which is the corresponding average value. For example, if $n=4$ and $C^s = \langle 7, 8, 8, 13 \rangle$, then the players' ranks are $\langle 1, 2.5, 2.5, 4 \rangle$, and if all players have the same completion time, then they all have rank $(n+1)/2$.

The primary objective of every player is to minimize its rank. The secondary objective is to minimize its completion time. Formally, player $j$ prefers profile $s'$ over profile $s$ if $rank_j(s') < rank_j(s)$ or $rank_j(s') = rank_j(s)$ and $C_j(s') < C_j(s)$. In this case, a deviation from profile $s$ to profile $s'$ is called beneficial for player $j$. A strategy profile $s$ is a {\em pure Nash equilibrium (NE)} if for all $j\in J$, job $j$ does not have a beneficial deviation.

For a strategy profile $s$, let $SC(s)$ denote the social cost of $s$. The social cost is defined with respect to some objective, e.g., the makespan, i.e., $C_{max}(s)=\max_{j\in J} C_j(s)$, or the sum of completion times, i.e., $\sum_{j \in J}C_j(s)$. It is well known that decentralized decision-making may lead to sub-optimal solutions from the point of view of the society as a whole. For a game $G$, let $P(G)$ be the set of feasible profiles of $G$. We denote by $OPT(G)$ the social cost of a social optimal solution, i.e., $OPT(G)=\min_{s \in P(G)} SC(s)$. We quantify the inefficiency incurred due to self-interested behavior according to the \emph{price of anarchy} (PoA) \cite{KP09}, and \emph{price of stability} (PoS) \cite{AD+08}. The PoA is the worst-case inefficiency of a pure Nash equilibrium, while the PoS measures the best-case inefficiency of a pure Nash equilibrium. 

\begin{definition}
\label{def:ineff}
Let $\G$ be a family of games, and let $G$ be a game in $\mathcal{G}$.
Let $\Eq(G)$ be the set of pure Nash equilibria of the game $G$. Assume that $\Eq(G) \neq \emptyset$.
\begin{itemize}
\item The {\em price of anarchy} of $G$ is the ratio between the
\emph{maximum} cost of a NE and the social optimum of
$G$, i.e.,
$\mbox{PoA}(G) = \max\limits_{s \in \Eq(G)} SC(s)/OPT(G)$.
The {\em price of anarchy} of $\mathcal{G}$
is $\mbox{PoA}(\mathcal{G}) = sup_{ G\in \mathcal{G}}\mbox{PoA}(G)$.
\item
The {\em price of stability} of $G$ is the ratio between the
\emph{minimum} cost of a NE and the social optimum of
$G$, i.e.,
$\mbox{PoS}(G) = \min\limits_{s \in \Eq(G)} SC(s)/OPT(G)$.
The {\em price of stability} of $\mathcal{G}$ is
$\mbox{PoS}(\mathcal{G}) = sup_{ G\in \mathcal{G}}\mbox{PoS}(G)$.
\end{itemize}
\end{definition}

The strategy profile of all players except player $j$ is denoted by $s_{-j}$, and it is convenient to denote a strategy profile $s$ as $s=(s_j, s_{-j})$.

\emph{Best-Response Dynamics} (BRD) is a natural method by which players proceed toward a NE via the following local search method:
Given a strategy profile $s$, the best response of player $j$ is $BR_j(s) = \argmin_{s'_j \in M} \langle rank_j(s'_j,s_{-j}), C_j(s'_j,s_{-j}) \rangle$; i.e., the set of strategies that maximize job $j$'s utility, fixing the strategies of all other players.
Player $j$ is said to be {\em suboptimal} in $s$ if it has a beneficial deviation (reduce its rank, or maintain its rank and reduce its completion time), i.e., if $s_j \not\in BR_j(s)$.
If no player is suboptimal in $s$, then $s$ is a NE.

Given an initial strategy profile $s^0$, a best response sequence from $s^0$ is a sequence $\langle s^0, s^1, \ldots \rangle$ in which for every $T=0,1,\ldots$ there exists a player $j \in J$ such that $s^{T+1} = (s'_{j}, s^{T}_{-j})$, where $s'_{j} \in BR_j(s^{T}_{-j})$.

Given a game $G$, the strategy profile graph of $G$ is a directed graph whose vertex set consists of all possible strategy profiles of $G$, and there is a directed edge $(s_1, s_2)$, if the profile $s_2$ can be obtained from profile $s_1$ by a best-response deviation of a single player.

A {\em deviator rule} is a function that, given a profile $s$, chooses a deviator among all suboptimal players in $s$.
The chosen player then performs a best response move (breaking ties arbitrarily). Given an initial strategy profile $s^0$ and a deviator rule $D$ we denote by $NE_{D}(s^0)$ the set of NE that can be obtained as the final profile of a BR sequence $\langle s^0,s^1, \ldots\rangle$, where for every $T \geq 0$, $s^{T+1}$ is a profile resulting from a deviation of $D(s^T)$.
Note that every BR-sequence corresponds to some path in the strategy profile graph. Therefore, the analysis of this graph is a major tool in understanding BRD convergence and the quality of possible BRD outcomes.

As we are going to show, a game with rank-based utilities may not have a NE. In order to analyze the quality of the profiles to which natural dynamics may converge, we use the concept of a {\em Sink Equilibrium}, introduced in~\cite{GMV05}. A sink equilibrium is a strongly connected component with no outgoing edges in the strategy profile graph. A sink equilibrium always exists. Thus, even in games that do not admit a NE, we can still analyze the expected quality of a steady state of the game. 



The definition of a sink equilibrium in~\cite{GMV05} refers to a random walk in the strategy profile graph. When BRD is applied with a specific deviator rule, the choice of the deviating player is deterministic, however, we can still have several different BR-sequences, depending on the initial profile and the tie breaking applied by a player in case it has more than one best-response move.
For a game $G$ and a deviator rule $D$, a sink equilibrium of $(G,D)$ is a strongly connected component with no outgoing edges in the strategy profile graph, to which $G$ may converge when $D$ is applied.

The following definitions extend the definitions from \cite{GMV05} to sink equilibria reached via a specific deviator rule.

\begin{definition}
\label{def:ineff_3}
Let $\mathcal{Q}(G,D)$ be the set of sink equilibria of $G$ reached by applying BRD with deviator rule $D$. For a sink $Q \in \mathcal{Q}(G,D)$, let $f_Q: Q \to \mathbb{R}^+$ be the steady state distribution of a BR-sequence over states in $Q$.
\begin{itemize}
\item The (expected) social cost of a sink equilibrium $Q \in \mathcal{Q}(G,D)$, denoted by $SC(Q)$, is the expected social cost of the states in a BR-sequence that reaches $Q$, i.e., $SC(Q) = \Sigma_{s \in Q} f(s) \cdot SC(s)$.
\item The {\em price of sinking} of $G,D$ is the ratio between the
\emph{maximum} cost of a sink equilibrium and the social optimum of $G$, i.e.,
$\mbox{PoSINK}(G,D) = \max\limits_{Q \in \mathcal{Q}(G,D)} SC(Q)/OPT(G)$. The {\em price of sinking} of $\mathcal{G},D$
is $\mbox{PoSINK}(\mathcal{G},D) = sup_{ G\in \mathcal{G}}\mbox{PoSINK}(G,D)$.
\end{itemize}
\end{definition}


\subsection{Rank-Based vs. Cost-Based Utilities}
\label{sec:util}
The crucial difference between rank-based utilities and cost-based utilities is demonstrated already in the following simple game with two identical machines $\{M_1, M_2\}$ and two unit-length jobs $J=\{a,b\}$ (see Figure~\ref{fig:NoNE}). Assume that in both machines $a \prec b$. It is easy to see that with rank-based utilities, the game has no NE profile. Specifically, let $s$ be a profile in which both jobs are assigned to the same machine. We have that $rank_a(s)=1, rank_b(s)=2$, thus job $b$ would benefit from deviating to the other machine. In the resulting profile, $s'$, we have that $rank_a(s')=rank_b(s')=1.5$. Any profile in which the jobs are on different machines, and in particular $s'$, is not stable as well, since both jobs have the same completion time and rank, therefore, job $a$ would benefit from joining the machine of job $b$. Such a deviation will not affect $a$'s completion time, and will delay job $b$. 
Note that in the traditional model, in which jobs only care about their completion time, any balanced schedule is a NE.

\begin{figure}[ht]
\centering
\includegraphics[width=0.7\textwidth]{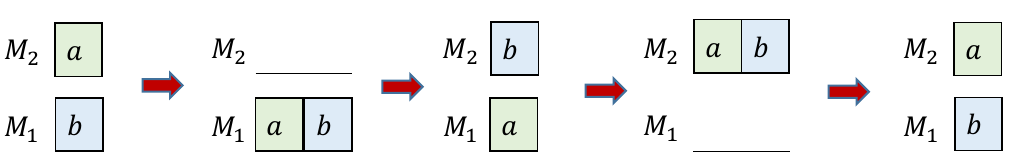}
\caption{\normalfont{
A BR-sequence in a game with two unit-length players and no NE.}}
\label{fig:NoNE}
\end{figure}

On the other hand, if the two machines have different priority lists, for example if $a \prec_1 b$ (that is, $M_1$ prioritize $a$ over $b$), and $b \prec_2 a$ then the assignment in which job $a$ is on $M_1$, and job $b$ is on $M_2$ is a NE. Both jobs have rank $1.5$ and will have rank $2$ if they deviate.

The above example highlights the fact that games with rank-based utilities are significantly different from classical job-scheduling games; the competition creates interesting problems already with unit-length jobs, a class whose analysis in the competition-free setting is straightforward.

The next example demonstrates that in a game with rank-based utilities, jobs may perform non-intuitive beneficial migrations. Specifically, a job may reduce its rank even though it increases its cost (completion time). Consider an instance with two identical machines $\{M_1, M_2\}$ and three jobs $J=\{a,b,c\}$, where $p_b > p_a+p_c$. Assume that in both machines $a \prec b \prec c$. Consider a schedule, $s$, in which jobs $a,c$ are on one machine and job $b$ is on the other machine (see left schedule in Figure~\ref{fig:CostInc}). We have that $rank_a(s)=1, rank_b(s)=3, rank_c(s)=2$. Job $b$ can reduce its rank from $3$ to $2$ by deviating to $M_1$. This deviation increases its completion time from $p_b$ to $p_a+p_b$, but is beneficial in our model. It is not difficult to see that this game has no NE.

\begin{figure}[ht]
\centering
\includegraphics[width=0.5
\textwidth]{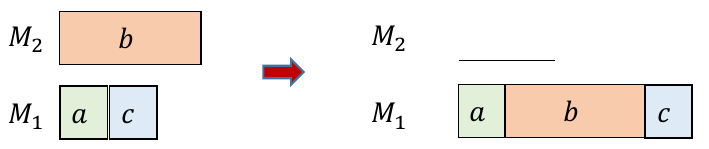}
\caption{\normalfont{
A cost-increasing rank-reducing deviation.}}
\label{fig:CostInc}
\end{figure}

\subsection{Related Work}

Scheduling games were initially studied in the setting in which each machine processes its jobs in parallel so that the completion time of each job depends on the total load on the machine. The corresponding papers analyze the inefficiency of selfish behavior with respect to the makespan. Czumaj and Vöcking~\cite{CzumajV07} gave tight bounds on the price of anarchy for related machines, whereas Awerbuch et al.~\cite{AART06} and Gairing et al.~\cite{GLM10} provided tight bounds for restricted machine settings. 
These tight bounds grow with the number of machines and motivated the study of coordination mechanisms, i.e., local scheduling policies, to reduce the price of anarchy.

Christodoulou et al.~\cite{CKN04} introduced coordination mechanisms and studied the price of anarchy with priority lists based on longest processing time (LPT) first. Immorlica et al.~\cite{ILMS09} generalized their results and studied several different scheduling policies. Other results were given by Cohen et al.~\cite{CDN11}, who proved that a Nash equilibrium is expected to exist for two unrelated machines with a random order. Also, Azar et al.~\cite{AJM08} showed that for unrelated machines with priorities based on the ratio of a job's processing time to its shortest processing time a Nash equilibrium need not exist; Lu and Yu \cite{LY12} suggested a mechanism that guarantees the existence of a Nash equilibrium; Kollias \cite{K13} showed that non-preemptive coordination mechanisms need not induce a pure Nash equilibrium. Vijayalakshmi et al.~\cite{RST21} considered a more general setting in which machines have arbitrary individual priority lists. The paper characterizes four classes of instances in which a pure Nash equilibrium is guaranteed to exist, and analyzes the equilibrium inefficiency for these classes. Caragiannis et al.~\cite{CF19} introduced {\em DCOORD} coordination mechanism with price of anarchy $O(\log m)$ and price of stability $O(1)$.

Best Response Dynamics (BRD) has been a significant area of study in the analysis of congestion games.
Research in this domain often examines the existence of NE, the convergence of BRD to these equilibria, and the inefficiencies introduced by selfish behavior. The papers \cite{EKM03,FT15,FST17,KBL13} analyze various deviator rules, such as Max-Weight-Job, Min-Weight-Job, and Max-Cost, and compared their effects on the convergence time and the solution quality in several classes of congestion games.

Research has also delved into scenarios where BRD does not converge to a pure Nash Equilibrium.
Goemans et al.~\cite{GMV05} introduced the concept of sink equilibria, to address situations where pure NE do not exist. They defined the {\em Price of Sinking} that quantifies the inefficiency of sink equilibria. 
Berger et al.~\cite{BFNR11} introduced the notion of dynamic inefficiency, examining the average social cost across an infinite sequence of best responses. They explored how different deviator rules, such as Random Walk, Round Robin, and Best Improvement, impact social costs in games lacking the finite improvement property.

Rank-based scheduling games where studied so far only on parallel machines (without priority lists). Ashlagi et al.~\cite{AKT08} presented a social context game with rank competition. The competition structure in their model is arbitrary and defined by a network. For the case of disjoint competition sets, that is, when the network is a collection of disjoint cliques, the paper shows that a NE may not exist and is guaranteed to exist in a game with identical resources. Immorlica et al.~\cite{ILMS09} consider a model with arbitrary competition structure, in which players' utility combine their payoff and ranking. General ranking games are studied in Brandt et al.~\cite{BFHS09}, where it is shown that computing a NE is NP-complete in most cases. In Goldberg et al.~\cite{GGKV13} a player's utilization combines its rank with the effort expended to achieve it. Rosner and Tamir introduced in~\cite{RT23} a scheduling game with rank-based utilities (SRBG), where the players are partitioned into competition sets, and the goal of every player is to perform well relative to its competitors. These works show that the analysis of games with rank-based utilities tends to be very different from the analysis of classical games.

Our model combines the study of coordination mechanism with rank-based utilities. This combination was not considered in the past.

\subsection{Our Results}
As demonstrated in Section~\ref{sec:util}, even in simple games with just two jobs and two machines, a pure Nash equilibrium may not exist. The natural problem of deciding whether a given game has a pure NE is central in our study. In Section~\ref{sec:hardness}, we establish that this problem is NP-complete, even when job lengths are restricted to $\{1, 2\}$. On the positive side, we identify several classes of instances where a Nash equilibrium is guaranteed to exist or where the decision problem can be solved efficiently.

The following classes of games are analyzed in our work:
\begin{itemize}
    \item $\G^{P}$ - games played on identical (unit-rate) machines.
    \item $\G^{Q}$ - games played on related machines.
    \item $\G^{global}$ - games with a global priority list, $\pi$. That is, for every machine $M_i \in M$, it holds that $\pi_i=\pi$.
    \item $\G_{unit}$ - games with unit-length jobs.    
\end{itemize}

Some of our results refer to games in the intersection of some classes. For example, a game $G \in \G_{unit}^{Q2, global}$ is played on two related machines with a global priority list and unit-length jobs.


In Section~\ref{sec:prel} we propose two simple greedy algorithms for computing a schedule: one for identical machines and one for related machines. Next, we provide necessary and sufficient conditions for the stability of the resulting schedule. A run of the algorithm may require $O(n)$ tie-breaking steps. We show that the question of deciding whether a game has a NE is reduced to the problem of deciding whether there exists a run of the algorithm whose output fulfills the stability conditions. Some of our results below are based on showing that even though the number of potential outputs is exponential, it is possible to trace them efficiently.

In Section~\ref{sec:identical_machines} we study equilibrium existence and computation in games played on identical machines. We first introduce a coordination mechanism denoted {\em Inversed-Policies}, that ensures the existence of a NE on any game with two machines. Next, we characterize games in $\G_{unit}^{P,global}$ that admit a NE, and provide a linear time algorithm for games in $\G_{unit}^{P2}$.

In Section~\ref{sec:related} we study equilibrium existence and computation in games played on two related machines. We provide an exact characterization of games in $\G_{unit}^{Q2}$ that admit a NE, distinguishing between global and machine-dependent priority lists, and between games where the ratio of the machines' rates is a rational or irrational number.

Focusing on games that admit a Nash equilibrium (NE), we turn in Section~\ref{sec:inefficiency} to examine how competition impacts the efficiency of these equilibria. Our efficiency metric is the makespan, defined as the maximum completion time across all jobs. We first verify that the tight bounds established in prior work for traditional scheduling games without competition remain valid when rank-based utilities are introduced.
Specifically, we find that the PoA and PoS of $\G^{P}$ are both $2 - \frac{1}{m}$, a result which also extends to games in $\G^{P2}$ with {\em Inversed-Policies} mechanism. For related machines, we show that the PoA and PoS of a game in $\G^{Q2}$ depend on the machine's speed ratio $r$: they are equal to $r + 1$ when $r \leq \frac{\sqrt{5} - 1}{2}$ and $\frac{r+2}{r+1}$ when $r > \frac{\sqrt{5} - 1}{2}$. Notably, when $r = \frac{\sqrt{5} - 1}{2}$, the system reaches its highest inefficiency, with both PoA and PoS equal $\frac{\sqrt{5} + 1}{2}$.
%
%

While the bounds are similar to those in traditional scheduling games, our analysis reveals instances where competition either mitigates or worsens the equilibrium inefficiency. Specifically, we present cases where a game without competition attains the worst possible price of stability (PoS) in its class, but with the introduction of competition, the price of anarchy (PoA) improves to $1$. Conversely, we present games where introducing competition degrades the equilibrium quality.

For games in $\G_{unit}^{P, global}$ that do not admit a NE, we propose a deviator rule and show that when it is applied, the price of sinking (PoSINK) is as low as $1 + \frac{1}{2 \cdot \ceil{\frac{n}{m}}}$. On the other hand, for games in $\G^{Q2}$ that do not admit a NE, we show that the price of sinking is not bounded by a constant.

Finally, in Section~\ref{sec:competition_classes} we analyze games with different competition classes, where the job set is partitioned into competition sets, and players aim to perform well within their set. We demonstrate that varying partitions into competition classes can ensure or prevent the existence of a pure NE. On the positive side, for games in $\G^{P2}$ with {\em Inversed-Policies}, we prove that a NE exists for any partition of $J$. Additionally, we analyze a modified model, in which priority lists are determined by a permutation of the competition sets rather than a full permutation of $J$. Assuming seniority-based priority within each set, we show that in this model a NE must exist and BRD converges to a NE.


\section{Computational Complexity}
\label{sec:hardness}

In this section, we show that the problem of deciding whether a game has a NE is NP-complete already in the case that job lengths are in $\{1,2\}$. Formally,

\begin{theorem}
\label{thm:hardUnit}
Given a coordination mechanism with rank-based utilities, the problem of deciding whether the game has a NE is NP-complete, even when the machines are identical and for all jobs $p_j \in \{1, 2\}$.
\end{theorem}

\begin{proof}
Given a game and a profile $s$, verifying whether $s$ is a NE can be done by checking for every job whether its current assignment is also its best-response, therefore the problem is in NP.

We show a reduction from $3$-bounded $3$-dimensional matching ($3$DM$3$). Our proof is based on an idea used in the hardness proof in~\cite{LST90}.
The input to the $3$DM-$3$ problem is a set of triplets $T \subseteq X \times Y \times Z$,
where $|T|\ge n$ and $|X|=|Y|=|Z|=n$. The goal is to decide whether $T$ has a $3$-dimensional matching of size $n$,  i.e., there exists a subset $T' \subseteq T$, such that $|T'|=n$, and every element in $X \cup Y \cup Z$ appears exactly once in $T'$.
$3$DM-$3$ is known to be NP-complete \cite{Kann91} even if the number of occurrences of every element of $X \cup Y \cup Z$ in $T$ is at most $3$.
We assume further that the number of occurrences of every element of $X \cup Y \cup Z$ in $T$ is exactly $2$ or $3$. 
This is w.l.o.g., since if an element appears in a single triplet, $t \in T$, then $t$ must be included in the solution, and the instance can be adjusted accordingly  -- by removing triplets containing the two other elements in $t$.

Given an instance $T$ of $3$DM-$3$, we construct the following game, $G_T$.
There are $m=|T|$ machines, $M_1,M_2,\ldots,M_{|T|}$. The $|T|$ machines correspond to the $|T|$ triplets.  The set of jobs consists of:
\begin{enumerate}
    \item $2n$ unit-length {\em element jobs} - one for each element in $Y \cup Z$.
    \item For every $1 \le i \le n$, we denote by {\em triples of type $i$} the triples that contain $x_i$. Let $\tau_i$ be the number of triples of type $i$. There are $\tau_i -1$ dummy jobs of type $i$. All dummy jobs have length $2$. Let $D_i$ denote the set of dummy jobs of type $i$, and let $D = \cup_i D_i$. Note that $|D|=\sum_i (\tau_i-1) = |T|-n$.
    \item Two sets $U=\{u_1,\ldots,u_{|T|}\}$ and $V=\{v_1,\ldots,v_{|T|}\}$, each with $|T|$ dummy unit-length jobs.
\end{enumerate}

All together there are $n+3|T|$ jobs, out of which $2(n+|T|)$ have unit length and $|T|-n$ have length $2$. Thus, the total processing time of jobs in $G_T$ is $2(n+|T|)+2(|T|-n) = 4|T|=4m$.

The heart of the reduction lies in determining the priority lists. Every machine corresponds to a triplet, and the jobs corresponding to this triplet are prioritized over other jobs. The idea is that if a $3$DM-$3$ matching exists, then for every $1 \le i \le n$, the $\tau_i$ machines corresponding to triplets of type $i$ will be assigned all the jobs originated from the triplet to which $x_i$ belongs, Specifically, $D_i$ as well as one job from $Y$ and one job from $Z$. Additionally, to make the schedule stable, each machine will be assigned also one $U$-job and one $V$-job, such that the load on every machine is exactly $4$, and the schedule is a NE.

The priority lists are defined as follows. When the list includes a set, it means that the set elements appear in a fixed arbitrary order. 
For every triplet $t_{\ell}=(x_i,y_j,z_k) \in T$, the priority list of $M_{\ell}$ is $$\pi_{\ell}=(D_i, y_j,z_k,U,v_{\ell}, V\setminus \{v_{\ell}\}, D \setminus D_i,  Y\setminus\{y_j\}, Z\setminus\{z_k\}).$$

Observe that every machine prioritizes the jobs associated with its elements, and also a different $V$-job among the dummy jobs in $V$.

In order to complete the proof we prove the following two claims that relate the existence of a NE in the game $G_T$ to the existence of a perfect matching in the $3$DM-$3$ instance $T$. We first show that if the $3$DM-$3$ instance has a perfect matching, then the game induced by our construction has a pure Nash equilibrium.

\begin{claim}
\label{clm:1:hard}
If a $3$D-matching of size $n$ exists in $T$, then the game $G_T$ has a NE.
\end{claim}

\begin{proof}
Let $T' \subseteq T$ be a matching of size $n$. 
Consider a triplet $t_\ell=(x_i, y_j,z_k)$ in $T'$. Note that $t_\ell$ is of type $i$. Assign $y_j$ and $z_k$ on $M_{\ell}$, and the $\tau_i-1$ jobs of $D_i$ on the other machines of type $i$.
Finally, assign on each machine one $U$-job, and the $V$-job it prioritizes.
See an example in Figure~\ref{fig:hard2}.

\begin{figure}[h]
  \centering
  \includegraphics[width=0.88\textwidth]{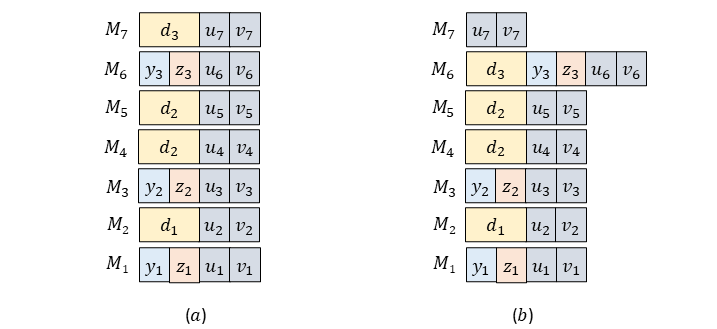}
  \caption{Let $T=\{t_1=\{x_1,y_1,z_1\}, t_2=\{x_1,y_2,z_3\}, t_3=\{x_2,y_2,z_2\}, t_4=\{x_2,y_1,z_2\}, t_5=\{x_2,y_3,z_1\}, t_6=\{x_3,y_3,z_3\}, t_7=\{x_3,y_1,z_1\}\}$. Note that $\{t_1,t_3,t_6\}$ is a perfect matching. (a) a NE profile. $(b)$ a non-beneficial deviation of $d_3$.}
  \label{fig:hard2}
\end{figure}

Let us verify that the resulting assignment is a NE. 
Recall that $|D|= |T|-n$. Thus, $|T|-n$ machines are assigned a job of length $2$ followed by two unit jobs, and $n$ machines are assigned four unit jobs.
The $Y$-jobs all complete at time $1$ and have rank $\frac{n+1}{2}$. The $Z$-jobs all complete at time $2$ and have rank $n+\frac{|T|+1}{2}$.
A migration of a $Y$-job or a $Z$-jobs will cause it to be placed after some job in $D$ or after the unit-length dummy jobs. Therefore, it will delay its completion time and will worsen its rank, implying that such a migration is not beneficial.

The completion time of each of the $D$-jobs is $2$, and its rank is $n+\frac{|T|+1}{2}$. A migration of a job in $D_i$ to a machine of type $i' \neq i$ will cause to to start after all the jobs on this machine and is clearly not beneficial. A migration of a job in $D_i$ to a different  machine of type $i$ may keep its completion time (depending on the internal order of the $D_i$-jobs on this machine), but will surely hurt its rank, as one $U$-job will now have a lower completion time. A migration to the machine of type $i$ that processes $y_j$ and $z_k$, will keep its completion time and its rank: Job $y_j$ does not complete before it anymore, however, the dummy $U$-job on its machine will now complete before it. Job $z_k$ does not complete at time $2$ anymore, however, the dummy $V$-job on its machine now complete it time $2$. Thus, its rank will remain $n+\frac{|T|+1}{2}$ (see Figure~\ref{fig:hard2}(b)).

The completion time of all the jobs in $U$ is $3$. A migration of a $U$-job will either delay its completion time and worsen its rank, or keep its completion time and rank. Finally, since every $V$-job is assigned to a machine that prioritizes it, a migration of a $V$-job, will delay its completion time and worsen its rank. We conclude that the profile is a NE.
\end{proof}

The next claim shows that if the $3$DM-$3$ instance does not have a perfect matching, then our construction guarantees that a no-NE subgame is triggered, and $G_T$ has no NE.

\begin{claim}
\label{clm:2:hard}
If the game $G_T$ has a NE profile, then a $3$D-matching of size $n$ exists.
\end{claim}

\begin{proof}
Let $s$ be a NE schedule of $G_T$. We show that $s$ induces a $3$D-matching of size $n$.
Since for all $i$ the $D_i$-jobs are first in $\tau_i$ priority lists, and $|D_i| = \tau_i - 1$, the jobs in $D_i$ are processed first on $\tau_i - 1$ machines. Also, the load on the machines is balanced. Specifically, the load on each machine is exactly $4$. Otherwise, a job with maximal completion time can deviate to a machine with load less than $4$ and reduce both its completion time and rank. 
Thus, in every machine in $s$, the jobs with completion times $3$ and $4$ are unit jobs. For $k \in \{3,4\}$ we denote by {\em layer $k$} the set of unit jobs with completion time $k$. 

Next note that the total length of jobs that precede the $U$-jobs in the priority lists is $2m$. Therefore, in every  NE profile, every $U$-job completes not later than at time $3$, as otherwise, some $U$ has a beneficial deviation.

We show that the $4^{th}$ layer includes only $V$-jobs,
such that every job in $V$ is assigned to the machine that prioritizes it, that is, $v_{\ell}$ is processed $4^{th}$ on $M_{\ell}$. 
Assume by contradiction that some machine $M_{\ell}$ does not process $v_{\ell}$ in the $4^{th}$ layer. By the discussion above, the jobs of $D$ and $U$ are not in the $4^{th}$ layer.
If $M_{\ell}$ processes in the $4^{th}$ layer a $V$-job, $v_{\ell'}$ for $\ell' \neq \ell$, then the no-NE game between $v_{\ell'}$ and the job in the $4^{th}$ layer on $M_{\ell'}$ comes to life, and $s$ is not a NE.
If $M_{\ell}$ processes some element-job, say $w \in Y \cup Z$, in the $4^{th}$ layer, we distinguish between two cases: $(i)$ if $w \not \in t_{\ell}$, then the no-NE game between any $V$-job on another machine and job $w$ comes to life, as demonstrated in Figure~\ref{fig:hardNoNE}. 
$(ii)$ If $w \in t_{\ell}$, recall that we assume that every element belongs to at least two triplets. Let $\ell' \in T$ be a different triplet to which $w$ belongs. Job $w$ has a beneficial deviation to $M_{\ell'}$, since it will be placed before the job that is currently in the $4^{th}$ slot on $M_{\ell'}$. 
We conclude that the $4^{th}$ layer consists of $V$-jobs. The same argument can be used to show that in every NE, the $3$rd layer consists of  $U$-jobs. Specifically, any element-job in the $3^{rd}$ layer has a beneficial deviation.

Consider the $\tau_i$ machines of type $i$. The $\tau_i-1$ dummy jobs of $D_i$ are assigned first on $\tau_i-1$ of these machines, as otherwise, $s$ is not a NE. 
Since each of these machines is also assigned jobs from $U$ and $V$ in the $3^{rd}$ and $4^{th}$ layer respectively, and the makespan is $4$, it is not assigned any element jobs from $Y \cup Z$.
The above observation is valid for $\sum_i (\tau_i-1) = |T|-n$ machines. Therefore, it must be that the jobs of $Y$ and $Z$ are processes on $n$ machines; every pair on the first two layers. The priority lists dictate that such an assignment is stable only if for every $i$, one machine, $M_{\ell}$, of type $i$, corresponding to a triplet $t_{\ell}=(x_i,y_j,z_k)$ is assigned the element jobs $y_j$ and $z_k$. Thus, a $3$D matching is induced.
\end{proof}

\begin{figure}[h]
  \centering
  \includegraphics[width=0.88\textwidth]{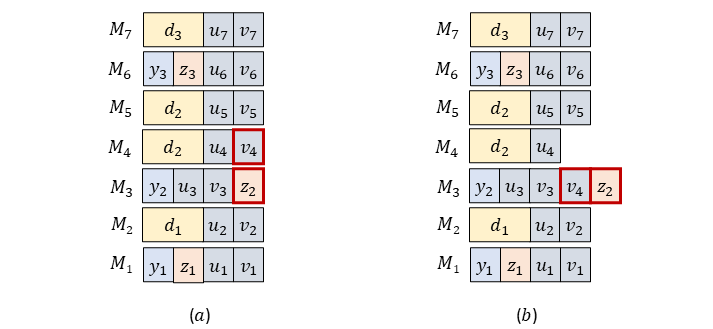}
  \caption{Let $T=\{t_1=\{x_1,y_1,z_1\}, t_2=\{x_1,y_2,z_3\}, t_3=\{x_2,y_2,z_1\}, t_4=\{x_2,y_1,z_2\}, t_5=\{x_2,y_3,z_1\}, t_6=\{x_3,y_3,z_3\}, t_7=\{x_3,y_1,z_1\}\}$. Note that  a perfect matching does not exist. (a) a possible profile. $(b)$ a beneficial deviation of $v_4$.}
  \label{fig:hardNoNE}
\end{figure}

The statement of Theorem~\ref{thm:hardUnit} follows from the two claims. 
\end{proof}

\section{Greedy Algorithms and Stability Conditions}
\label{sec:prel}

In this section we present greedy scheduling algorithms and provide sufficient and necessary conditions for the stability of their output. These algorithms will be analyzed further in our work. 

\subsection{Identical Machines}
The first algorithm is for identical machines. It assigns the jobs greedily one after the other in non-decreasing order of starting time.

\begin{algorithm}[H]
\caption{ - Computing a schedule $s$, identical machines}
\label{alg:AlgNext}
\begin{algorithmic}[1]
\STATE For $1 \le i \le m$, set $L_i(s) = 0$.
\REPEAT
\STATE Let $i^{\star} = \argmin_i ~ (L_i(s))$, breaking ties arbitrarily.
\STATE Assign on machine $i^{\star}$ the first unassigned job in $\pi_{i^{\star}}$, $j$.
\STATE $L_{i^{\star}}(s) = L_{i^{\star}}(s) + p_j$.
\UNTIL {all jobs are scheduled}
\end{algorithmic}
\end{algorithm}

An important property of Algorithm~\ref{alg:AlgNext} is that, as shown in \cite{CQ12}, it generates a schedule in which no job can decrease its completion time, that is, a schedule stable against cost-reducing deviations. However, as demonstrated in Section~\ref{sec:util}, such a profile might not be a NE in the presence of competition.

Algorithm~\ref{alg:AlgNext} will be analyzed for various classes of games in this paper. 
The following discussion refers to the class $\G_{unit}$ of games with unit-length jobs.
We begin by presenting a sufficient and necessary condition for having a beneficial deviation of a single job in a schedule in a game $G \in \G_{unit}$. Recall that a deviation is {\em cost-reducing} if it reduces the completion time of the deviating job.

\begin{claim}
\label{cl:beneficial}
Let $s$ be a schedule in a game $G \in \G_{unit}$, such that $s$ is stable against cost-reducing deviations. A job $j$ on machine $M_i$ has a rank-decreasing deviation iff  
\begin{enumerate}
\item Job $j$ is processed last on its machine, and 
\item There exists a job $j'$ on $M_z$, for which $C_{j'}(s) = C_{j}(s)$ and $j \prec_{z} j'$.
\end{enumerate}
\end{claim}

\begin{proof}
We start by showing that these two conditions are sufficient. Let $s$ be a schedule stable against cost-reducing deviations. Let $j$ be the last job on some machine, and let $j'$ be a job processed on machine $M_z$ such that $C_{j}(s) = C_{j'}(s)$ and $j \prec_{z} j'$. Note that $rank_s(j) = rank_s(j')$. A migration of $j$ to $M_{z}$ will keep its completion time, but will decrease its rank since one less job is tied with it and job $j'$ now completes after it.

Now we show that the conditions are necessary. Let $s$ be a schedule stable against cost-reducing deviations, and let $j$ be a job that has a beneficial migration from $M_{i}$ to $M_z$. Let $s'$ be the schedule resulting from this migration. Since $j$ cannot decrease its completion time, the only possible beneficial migration is rank-decreasing. Thus, there exists a job $j'$ on machine $M_{z}$ such that $C_{j'}(s) \leq C_{j}(s)$ and $C_{j'}(s') > C_{j}(s')$. Since $p_j = p_{j'}$, if $C_{j'}(s) < C_{j}(s)$, then $C_{j}(s') < C_{j}(s)$, contradicting the fact that $s$ is stable against cost-reducing deviations. Therefore, $C_{j'}(s) = C_{j}(s)$. Since $j$ and $j'$ are processed on $M_{z}$ in $s'$, it follows that $j \prec_{z} j'$.
Lastly, since $C_j(s') = C_j(s)$, job $j$ must be the last job processed on $M_i$ in $s$, as otherwise, if any other job is processed after $j$, the number of jobs tied with $j$ in $s'$ would remain the same as in $s$, making the migration to $M_z$ non-beneficial.
\end{proof}


Based on the above claim, we now describe the conditions under which a schedule produced by Algorithm~\ref{alg:AlgNext} is a NE for a game $G \in \G_{unit}$. Let $n = \ell \cdot m + c$ for $0 \le c < m$. Note that when the algorithm is executed on an instance with unit-length jobs, then in the resulting schedule, $m - c$ machines have load $\ell = \floor{n/m}$, and $c$ machines have load $\ell + 1 = \ceil{n/m}$.

\begin{theorem}
Let $s$ be a schedule produced by Algorithm~\ref{alg:AlgNext} for a game $G \in \G_{unit}$. Schedule $s$ is a NE iff the two following conditions hold:    
\begin{enumerate}
\item Let $P$ be the set of $c$ jobs with completion time $\ell + 1$. If $j \in P$ is processed on machine $i$, then $j$ has the highest priority in $\pi_i$ among the jobs in $P$.
\item Let $P_1$ be the set of $m - c$ jobs with completion time $\ell$ that are processed last on their machines, and $P_2$ be the set of $c$ jobs with completion time $\ell$ that are not processed last on their machines. If $j \in P_1$ is processed on machine $i$, then $j$ has the highest priority in $\pi_i$ among the jobs in $P_1$. If $j \in P_2$ is processed on machine $i$, then $j$ has higher priority in $\pi_i$ than any $j' \in P_1$.
\end{enumerate}
\end{theorem}

\begin{proof}
Let $s$ be a NE schedule produced by Algorithm~\ref{alg:AlgNext}. Since all jobs in $P$ share the same completion time, and similarly, all jobs in $P_1$ and $P_2$ share their respective completion times, conditions $1$ and $2$ follow directly from Claim~\ref{cl:beneficial} (see example in Figure~\ref{figP}). Specifically, in $P$, each job is assigned to the machine where it has the highest priority among the jobs with the same completion time, and similarly, the prioritization within $P_1$ and $P_2$ ensures no beneficial deviations.

For the other direction, let $s$ be a schedule in which both conditions are satisfied. Since no job can decrease its completion time, it follows from Claim~\ref{cl:beneficial} that only jobs that are processed last on their machines, i.e., jobs in $P$ or $P_1$, might potentially benefit from migrating. By condition $1$ and Claim~\ref{cl:beneficial}, no job in $P$ has a beneficial migration. Similarly, by condition $2$ and Claim~\ref{cl:beneficial}, no job in $P_1$ has a beneficial migration.
\end{proof}

\begin{figure}[h]
\centering
\includegraphics[width=1\textwidth]{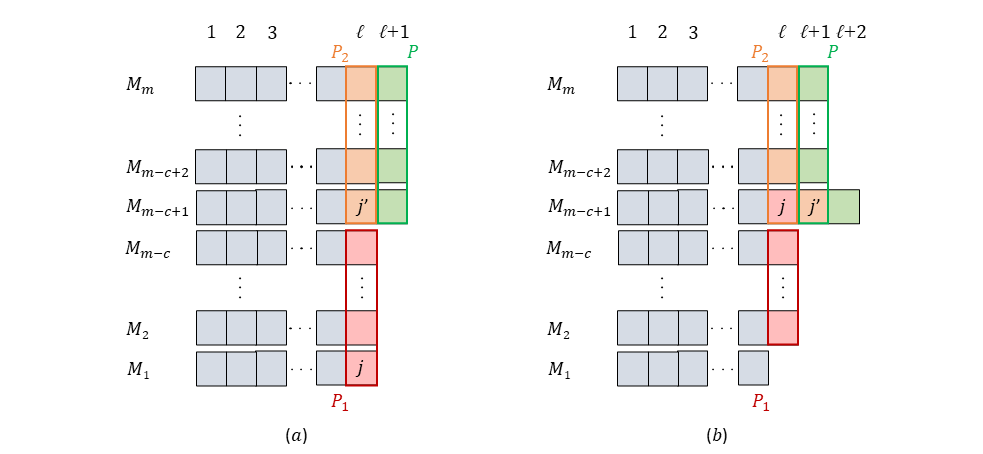}
\caption{(a) The sets $P,\ P_1$ and $P_2$ in $s$. (b) A beneficial migration of job $j \in P_1$ from $M_1$ to $M_{m - c +1}$, assuming condition $2$ is not satisfied and $j \prec_{m - c + 1} j'$.}
\label{figP}
\end{figure}

\subsection{Related Machines}
\label{sec:alg2}
The second algorithm is for related machines. It assigns the jobs greedily one after the other according to their expected {\em completion time}.

\begin{algorithm}[H]
\caption{ - Computing a schedule $s$, related machines and unit-jobs} \label{alg:AlgNextRelated}
\begin{algorithmic}[1]
\STATE For $1 \le i \le m$, set $L_i(s)=0$. 
\REPEAT
\STATE Let $i^{\star} = \argmin_i ~ (L_i(s) + 1) / r_i$, breaking ties arbitrarily.
\STATE Assign on machine $i^{\star}$ the first unassigned job in $\pi_{i^{\star}}$.
\STATE $L_{i^{\star}}(s) = L_{i^{\star}}(s) + 1$.
\UNTIL {all jobs are scheduled}
\end{algorithmic}
\end{algorithm}

Clearly, with identical machines, this algorithm is equivalent to Algorithm~\ref{alg:AlgNext}. We refer to Algorithm~\ref{alg:AlgNext} when analyzing identical machines, since its greedy choice is more intuitive.

Section~\ref{sec:related} considers the class $\G_{unit}^{Q2}$ of games with unit-jobs and two related machines, $M_1$ and $M_2$, with rates $r_1 = 1$ and $r_2 = r \leq 1$. 
We now establish a condition regarding the stability of a profile produced by Algorithm~\ref{alg:AlgNextRelated} for this class.
As shown in \cite{RST21}, Algorithm~\ref{alg:AlgNextRelated} produces a schedule in which no job can decrease its completion time. However, with rank-based utilities, such a schedule is not necessarily stable.

Given a schedule $s$, for $i \in \{1,2\}$, let $j_i$ be the last job on $M_i$ in $s$.

\begin{theorem}
\label{thm:NE_related}
Let $s$ be a schedule of a game $G \in \G_{unit}^{Q2}$ produced by Algorithm~\ref{alg:AlgNextRelated}. Schedule $s$ is a NE iff one of the three following conditions holds:
\begin{enumerate}
    \item $L_1(s) < L_2(s) / r$.
    \item $L_1(s) = L_2(s) / r$, and each of the two last jobs is prioritized on its machine over the other last job. Formally, $j_1 \prec_1 j_2$ and $j_2 \prec_2 j_1$. 
    \item $L_1(s) > L_2(s) / r$, and either no job on $M_1$ has completion time $C_{j_2}(s)$, or for the job $j'$ such that $C_{j'}(s)=C_{j_2}(s)$ it holds that $j' \prec_1 j_2$.
\end{enumerate}
\end{theorem}

\begin{proof}
We start by showing that each of these conditions ensures no job satisfies the conditions outlined in Claim~\ref{cl:beneficial}.
First, if the slow machine is more loaded, then clearly, $C_{j_2}(s)$ is not equal to the completion time of any job on $M_1$. Also, $C_{j_1}(s)$ is during the interval in which $j_2$ is processed on $M_2$. Second, if $L_1(s) = L_2(s) / r$ and $j_1 \prec_1 j_2$ and $j_2 \prec_2 j_1$, then none of the two jobs has a beneficial deviation, and no other job is last on its machine. Finally, if the third condition is valid, then as required, no last job has a same-rank job on the other machine, such that a deviation is beneficial.

If none of the above conditions hold, then, there must be a job $j$ that is last on its machine, and has a same-rank job on the other machine, such that the conditions for a beneficial deviation given in Claim~\ref{cl:beneficial} are met, and $s$ is not a NE.
\end{proof}


\section{Identical Machines - Equilibrium Existence and Computation}
\label{sec:identical_machines}
In this section we analyze games with identical machines, i.e., $\forall M_i \in M, r_i = 1$.

\subsection{Machine-Dependent Priority Lists and Arbitrary Job Lengths}

For the case of $m = 2$ identical machines, ($M_1, M_2$), with machine-dependent priority list and arbitrary job lengths, Christodoulou et al. \cite{CKN04} introduced the mechanism {\em Increasing-Decreasing} with tiebreak based on the lexicographic order of the jobs. We generalize their result and introduce the {\em Inversed-Policies} coordination mechanism, in which the priority lists of $M_1$ and $M_2$ are inversed. That is, $\forall j = 1, \dots, n$ \xspace \xspace $\pi_1(j) = n - \pi_2(j) + 1$. Note that the order is independent of the job lengths.

We show that a game with {\em Inversed-Policies} has a NE, and that Algorithm~\ref{alg:AlgNext} produces a NE profile.



\begin{theorem}
\label{Thm:InvPol_Convergence}
With {\em Inversed-Policies}, Algorithm~\ref{alg:AlgNext} produces a Nash equilibrium.
\end{theorem}

\begin{proof}
We show that upon termination of the algorithm, no job has a beneficial deviation. The proof is by induction on the number of jobs assigned in any prefix of the assignment. That is, for any $k \geq 1$, the partial schedule of the first $k$ jobs is stable.

Assume, w.l.o.g., that $\pi_1 = \langle j_1, j_2, \dots, j_{n-1}, j_n \rangle$ and $\pi_2 = \langle j_n, j_{n-1}, \dots, j_2, j_1 \rangle$. The base case considers the first assignment of a job in the algorithm, w.l.o.g., $j_1$ on $M_1$. For this partial schedule, we have that $rank_{j_1} = 1$ and $C_{j_1} = p_1$. Therefore, a deviation to $M_2$ is not beneficial.

Assume that the partial schedule of the first $k-1$ jobs is stable. Let $a$ and $b$ denote the number of jobs assigned so far to $M_1$ and $M_2$, respectively, such that $a + b = k - 1$. Consider the assignment of the $k^{th}$ job. Denote by $L_i$ the load on $M_i$ in the partial schedule. Assume first that $L_1 < L_2$, implying that the first unassigned job in $\pi_1$, $j_{a+1}$, is assigned to $M_1$ (see Figure~\ref{fig:inv-pol}). Every job $j_i$ for $i \le a$ on $M_1$ begins its processing before $L_2$, and since the jobs already on $M_2$ are prioritized over $j_i$ in $\pi_2$, a migration from $M_1$ to $M_2$ cannot be beneficial. Every job $j_{n - i}$ for $i < b$ on $M_2$ begins its processing not later than $L_1$, since the load on $M_2$ was equal to or less than the load on $M_1$ at the beginning of the iteration in which $j_{n - b + 1}$ was assigned. Since the jobs already on $M_1$ are prioritized over $j_{n - i}$ in $\pi_1$, a migration from $M_2$ to $M_1$ cannot be beneficial. Therefore, all assigned jobs are stable. The analysis for $L_1 \ge L_2$ is similar.
\end{proof}

\begin{figure}[h]
  \centering
  \includegraphics[width=1\textwidth]{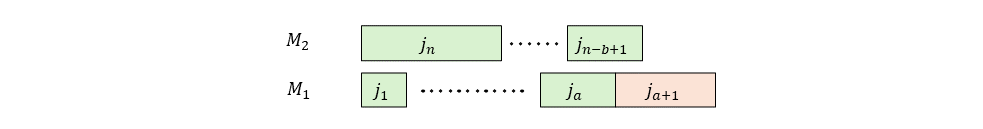}
  \caption{The $k^{th}$ iteration}
  \label{fig:inv-pol}
\end{figure}


\subsection{Global Priority List and Unit-Length Jobs}

In this section, we consider the class $\G_{unit}^{P,global}$, with unit-length jobs and a global priority list. We show that no NE exists if the game is played on more than two machines, and we give a simple characterization of instances on two machines that have a NE. Moreover, if a NE exists, then it can be computed efficiently.

\begin{theorem}
\label{thm:Unit_m2}
A game $G \in \G_{unit}^{P,global}$ has a NE iff $m=2$ and $n$ is odd.
\end{theorem}

\begin{proof}
We show that the condition is necessary and sufficient.

\begin{claim}
\label{clm:1:Unit_m2}
If $G$ has a NE, then $m = 2$ and $n$ is odd.
\end{claim}

\begin{proof}
Let $s$ be a NE schedule. It must be that $0 \leq \max_{i\in M} L_i(s) - \min_{i\in M} L_i(s) \leq 1$, as otherwise the last job processed in $s$ has a beneficial migration. If $m = 2$ and $n$ is even, or if $m > 2$, then there are at least two machines $M_1, M_2$ such that $L_1(s) = L_2(s)$. Let $j_1, j_2$ be the last jobs on $M_1, M_2$, respectively. Notice that $C_{j_1}(s) = C_{j_2}(s)$ and $rank_{j_1}(s) = rank_{j_2}(s)$. Assume w.l.o.g., that in the global priority list $j_1 \prec j_2$. A migration of $j_1$ to $M_2$ will improve its rank since one less job will share with $j_1$ the completion time $C_{j_1}(s)$, and one more job will have completion time $C_{j_1}(s) + 1$, contradicting the stability of $s$.
\end{proof}

\begin{claim}
\label{clm:2:Unit_m2}
If $m = 2$ and $n$ is odd, then $G$ has a NE.
\end{claim}

\begin{proof}
Let $n = 2\ell+1$. Assume that Algorithm~\ref{alg:AlgNext} is executed. Consider the schedule after the first $2\ell$ jobs are assigned.
Let $j_3$ be the last job in $\pi$, and let $j_1, j_2$ be the jobs with completion time $\ell$ on $M_1$ and $M_2$ respectively.
Assume w.l.o.g., that $j_1 \prec j_2$. We claim that an assignment of $j_3$ on $M_1$ will produce a NE.  In the resulting schedule, $s$, the load on $M_1$ is $\ell+1$ and the load on $M_2$ is $\ell$. Recall that $s$ is stable with respect to completion time. Therefore, all jobs that are not last on their machines are stable. A migration of $j_1$ to $M_2$ will keep its rank and completion time, and a migration of $j_2$ will increase its rank from $2\ell-\frac 1 2$  to $2\ell$. Therefore, $s$ is a NE.
\end{proof}

The statement of Theorem~\ref{thm:Unit_m2} follows from Claims~\ref{clm:1:Unit_m2} and \ref{clm:2:Unit_m2}.
\end{proof}

\subsection{Machine-Dependent Priority Lists and Unit-Length Jobs}



In this section, we consider the class $\G_{unit}^{P2}$, with unit-length jobs and two identical machines. We present a linear time algorithm for deciding whether a given game has a NE, and producing a NE if one exists.

\begin{theorem}
\label{thm:2MachinesUnitJobs}
For every $G \in \G_{unit}^{P2}$, it is possible to decide in linear time whether $G$ has a NE, and to produce a NE if one exists.
\end{theorem}

\begin{proof}
Recall Algorithm~\ref{alg:AlgNext}, which ensures that no job can decrease its completion time in the generated schedule. We utilize this property to prove that if $n$ is odd, the game $G$ has a NE.


\begin{claim}
\label{cl:odd_n_NE}
If $n$ is odd, any game $G \in \G_{unit}^{P2}$ has a NE, and a NE can be computed in time $O(n)$.
\end{claim}

\begin{proof}
Let $n = 2\ell + 1$. Assume that Algorithm~\ref{alg:AlgNext} is executed. Consider the schedule after the first $2\ell$ jobs are assigned. Let $j_1, j_2$ be the jobs with completion time $\ell$ on $M_1$ and $M_2$, respectively.
If $j_1 \prec_1 j_2$, assigning $j_3$ to $M_1$ results in a NE $s$, with $M_1$ having a load of $\ell+1$ and $M_2$ having a load of $\ell$. Since $s$ is stable with respect to completion time, all jobs with completion time less than $\ell$ are stable. A migration of $j_1$ to $M_2$ would maintain its rank and completion time, while a migration of $j_2$ would increase its rank from $2\ell-\frac{1}{2}$ to $2\ell$. Thus, $s$ is a NE. Conversely, if $j_2 \prec_1 j_1$, assigning $j_3$ to $M_2$ produces a NE $s'$, with $M_2$ having a load of $\ell+1$ and $M_1$ a load of $\ell$. The same stability reasoning applies: a migration of $j_2$ to $M_1$ maintains its rank and completion time, while a migration of $j_1$ would increase its rank from $2\ell-\frac{1}{2}$ to $2\ell$. Therefore, $s'$ is a NE.
\end{proof}


On the other hand, if $n$ is even, the existence of a NE is not guaranteed. We begin by showing that if a NE does exist, then Algorithm~\ref{alg:AlgNext} has the capability to generate it, or at least a schedule closely resembling it with identical {\em layers}, which we define as follows:

Let $s$ be an assignment of $J$. Denote the jobs assigned to machine $M_i$ by $j_{i1}, \dots, j_{ix}$ according to their order in $\pi_i$, with $x$ representing the total number of jobs assigned to $M_i$. For any $k \geq 1$, the $k^{th}$ layer of $s$ is denoted by $H_k^{(s)}$ and consists of the two jobs $\{j_{1k}, j_{2k}\}$.

The following claim establishes an important property of stable assignments.

\begin{claim}
\label{cl:anyNE}
Any stable assignment $s$ has a run of Algorithm~\ref{alg:AlgNext} that produces a schedule with the same layers as in $s$, and in which the jobs in the last layer are processed on the same machines as in $s$.
\end{claim}

\begin{proof}
Let $s$ be a NE, and recall that $n = 2\ell$. We observe that $s$ must be balanced, meaning that the number of layers is $\ell$. Consider the $k^{th}$ layer $H_k^{(s)}$ for some $1 \leq k \leq \ell$. For any $t > k$, clearly $j_{1k} \prec_1 j_{1t}$, and if $j_{2t} \prec_1 j_{1k}$ then $j_{2t}$ has a beneficial migration, in contradiction to the stability of $s$. Therefore, among the jobs that are not assigned in layers $1, \dots, k-1$, either $j_{1k}$ is first in $\pi_1$ or $j_{1k}$ is second and $j_{2k}$ is first. Similarly, among the jobs that are not assigned in layers $1, \dots, k-1$, either $j_{2k}$ is first in $\pi_2$ or $j_{2k}$ is second and $j_{1k}$ is first. In any of these options, there exists a run of Algorithm~\ref{alg:AlgNext} that creates a schedule in which layer $k$ consists of $j_{1k}$ and $j_{2k}$.

Let $j_{1\ell}$ and $j_{2\ell}$ denote the jobs in the last layer $\ell$. Since $s$ is a NE, it follows that $j_{1\ell} \prec_1 j_{2\ell}$ and $j_{2\ell} \prec_2 j_{1\ell}$. Therefore, there is an execution of the algorithm that assigns $j_{1\ell}$ and $j_{2\ell}$ on machines $1$ and $2$, respectively, as in $s$.
\end{proof}


We now establish the condition for a schedule produced by Algorithm~\ref{alg:AlgNext} to be stable.

\begin{claim}
\label{cl:stable_schedule}
Let $n = 2\ell$, and let $s$ be an assignment produced by Algorithm~\ref{alg:AlgNext}. Let $j_{1\ell}$ and $j_{2\ell}$ be the last jobs on $M_1$ and $M_2$ in $s$, respectively. Schedule $s$ is a NE iff $j_{1\ell} \prec_1 j_{2\ell}$ and $j_{2\ell} \prec_2 j_{1\ell}$.
\end{claim}

\begin{proof}
Assume first that $s$ is a NE. Thus, no job, and in particular $j_{1\ell}$ and $j_{2\ell}$, can benefit from deviation. Therefore, $j_{1\ell} \prec_1 j_{2\ell}$ and $j_{2\ell} \prec_2 j_{1\ell}$.

For the other direction, assume that $j_{1\ell} \prec_1 j_{2\ell}$ and $j_{2\ell} \prec_2 j_{1\ell}$. By Claim~\ref{cl:beneficial}, since $s$ is stable against completion time reducing deviations, only jobs $j_{1\ell}$ and $j_{2\ell}$ can potentially benefit from deviation. However, since $j_{1\ell} \prec_1 j_{2\ell}$ and $j_{2\ell} \prec_2 j_{1\ell}$, neither of them can benefit from deviation. Therefore, $s$ is a NE.
\end{proof}

From Claims~\ref{cl:anyNE} and~\ref{cl:stable_schedule}, we infer that by examining the last layers of all the potential schedules generated by Algorithm~\ref{alg:AlgNext}, we can determine the existence of a NE and even produce one if it exists. We show that at most two potential last layers exist, and we can identify them efficiently in linear time.

Recall that $n=2\ell$. For $1 \leq k \leq \ell$, let $\Gamma_k$ be the set of sets of jobs such that $S_k \in \Gamma_k$ if and only if $|S_k|=2k$ and there exists a run of Algorithm~\ref{alg:AlgNext} in which the jobs of $S_k$ are assigned on the first $k$ layers.

We show that for $1 \leq k \leq \ell$, $|\Gamma_k| \leq 2$. Moreover, if $|\Gamma_k|=2$ then the two sets are identical up to a single job.
For $1 \leq k \leq \ell$, denote by $\hat{\pi_i}^{(k)}$ the unassigned jobs in $\pi_i$, at the beginning of the $(2k-1)^{th}$ iteration of Algorithm~\ref{alg:AlgNext}, i.e., after assigning jobs in the first $k-1$ layers, and before assigning jobs in the $k^{th}$ layer. Note that $\hat{\pi_i}^{(1)} = \pi_i$, and that $|\hat{\pi_i}^{(k)}| = n - 2 \cdot (k-1)$. Denote by $\hat{\pi_i}^{(k)}(x)$ the job in the $x^{th}$ place in $\hat{\pi_i}^{(k)}$.

\begin{claim}
\label{cl:two_options}
For all $1 \leq k \leq \ell$, $|\Gamma_k| \le 2$. If $S_k^1, S_k^2 \in \Gamma_k$, then
\begin{enumerate}
    \item $|S_k^1 \cap S_k^2| = 2k-1$, and
    \item $\hat{\pi_1}^{(k)}(1) = \hat{\pi_2}^{(k)}(1)$, $\hat{\pi_1}^{(k)}(2) \neq \hat{\pi_2}^{(k)}(2)$, and
    \item denote $\hat{\pi_1}^{(k)}(2) = a_k$ and $\hat{\pi_2}^{(k)}(2) = b_k$, then $S_k^1 \setminus S_k^2 = \{a_k\}$, $S_k^2 \setminus S_k^1 = \{b_k\}$, w.l.o.g.
\end{enumerate}
\end{claim}


\begin{proof}
The proof is by induction on $k$.

Base Case ($k = 1$): Recall that $\hat{\pi_i}^{(1)} = \pi_i$. If $\pi_1(1) \neq \pi_2(1)$, then $\Gamma_1$ includes a single set $S_1 = \{\pi_1(1), \pi_2(1)\}$. If $\pi_1(1) = \pi_2(1)$ (say, job $x$), then the algorithm schedules job $x$ on one of the machines. If $\pi_1(2) = \pi_2(2)$ (say, job $y$), then again $\Gamma_1$ includes a single set $S_1 = \{x, y\}$. Otherwise, $\pi_1(2) \neq \pi_2(2)$. Denote $\pi_1(2),\ \pi_2(2)$ by $a_1,\ b_1$, respectively. Now, $\Gamma_1=\{S_1^1,S_1^2\}$, where $S_1^1 = \{x, a_1\}$ and $S_1^2 = \{x, b_1\}$. Note that $|S_1^1 \cap S_1^2| = 1$, $S_1^1 \setminus S_1^2 = \{a_1\}$ and $S_1^2 \setminus S_1^1 = \{b_1\}$.

Induction Step: If $|\Gamma_{k-1}|=1$, then the proof for $\Gamma_k$ is similar to the base case: Let $\Gamma_{k-1} = \{S_{k-1}\}$. If $\hat{\pi_1}^{(k)}(1) \neq \hat{\pi_2}^{(k)}(1)$ then $\Gamma_k$ includes a single set $S_k = S_{k-1} \cup \{\hat{\pi_1}^{(k)}(1), \hat{\pi_2}^{(k)}(1)\}$. Now assume $\hat{\pi_1}^{(k)}(1) = \hat{\pi_2}^{(k)}(1)$ (say, job $x$). If $\hat{\pi_1}^{(k)}(2) = \hat{\pi_2}^{(k)}(2)$ (say, job $y$), then again $\Gamma_k$ includes a single set $S_k = S_{k-1} \cup \{x, y\}$. Otherwise, $\hat{\pi_1}^{(k)}(2) \neq \hat{\pi_2}^{(k)}(2)$. Denote $\hat{\pi_1}^{(k)}(2)$ and $\hat{\pi_2}^{(k)}(2)$ by $a_k$ and $b_k$, respectively. Now, $\Gamma_k=\{S_k^1,S_k^2\}$, where $S_k^1 = S_{k-1} \cup \{x, a_k\}$ and $S_k^2 = S_{k-1} \cup \{x, b_k\}$. Note that $|S_k^1 \cap S_k^2| = 2k-1$, $S_k^1 \setminus S_k^2 = \{a_k\}$ and  $S_k^2 \setminus S_k^1 = \{b_k\}$. Thus, the claim holds.


We proceed by considering the case where $|\Gamma_{k-1}| = 2$. According to the induction hypothesis, $|S_{k-1}^1 \cap S_{k-1}^2| = 2 \cdot (k-1)-1$, $\hat{\pi_1}^{(k-1)}(1) = \hat{\pi_2}^{(k-1)}(1)$, $\hat{\pi_1}^{(k-1)}(2) \neq \hat{\pi_2}^{(k-1)}(2)$, and if $\hat{\pi_1}^{(k-1)}(2) = a_{k-1}$ and $\hat{\pi_2}^{(k-1)}(2) = b_{k-1}$, then $S_{k-1}^1 \setminus S_{k-1}^2 = \{a_{k-1}\}$, $S_{k-1}^2 \setminus S_{k-1}^1 = \{b_{k-1}\}$, w.l.o.g. We show that $|\Gamma_k| \le 2$.

Consider the different options for $\hat{\pi_1}^{(k)}, \hat{\pi_2}^{(k)}$:

\begin{enumerate}

    \item If $\hat{\pi_1}^{(k)}(1) = a_k$ and $\hat{\pi_2}^{(k)}(1) =b_{k-1}$, then $\Gamma_k$ includes a single set $S_k = S_{k-1}^1 \cup \{b_{k-1}, a_k\}$.

    \item If $\hat{\pi_1}^{(k)} = \langle b_{k-1}, x, \dots \rangle$ and $\hat{\pi_2}^{(k)} = \langle b_{k-1}, x, \dots \rangle$, then $\Gamma_k$ includes a single set $S_k = S_{k-1}^1 \cup \{b_{k-1}, x\}$.

    \item If $\hat{\pi_1}^{(k)} = \langle b_{k-1}, a_k, \dots \rangle$ and $\hat{\pi_2}^{(k)} = \langle b_{k-1}, b_k, \dots \rangle$, then $\Gamma_k=\{S_k^1,S_k^2\}$, where $S_k^1 = S_{k-1}^1 \cup \{b_{k-1}, a_k\}$, and $S_k^{2} = S_{k-1}^1 \cup \{b_{k-1}, b_k\}$.


    \item If $\hat{\pi_1}^{(k)}(1) = a_{k-1}$ and $\hat{\pi_2}^{(k)}(1) = b_k$, then $\Gamma_k$ includes a single set $S_k = S_{k-1}^2 \cup \{a_{k-1}, b_k\}$.

    \item If $\hat{\pi_1}^{(k)} = \langle a_{k-1}, x, \dots \rangle$ and $\hat{\pi_2}^{(k)} = \langle a_{k-1}, x, \dots \rangle$, then $\Gamma_k$ includes a single set $S_k = S_{k-1}^2 \cup \{a_{k-1}, x\} = S_k^2$.
    
    \item If $\hat{\pi_1}^{(k)} = \langle a_{k-1}, a_k, \dots \rangle$ and $\hat{\pi_2}^{(k)} = \langle a_{k-1}, b_k, \dots \rangle$, then $\Gamma_k=\{S_k^1,S_k^2\}$, where $S_k^1 = S_{k-1}^2 \cup \{a_{k-1}, a_k\}$, and $S_k^{2} = S_{k-1}^2 \cup \{a_{k-1}, b_k\}$.
\end{enumerate}

To conclude, considering all cases, we observe that $|\Gamma_k| \le 2$, and if $S_k^1, S_k^2 \in \Gamma_k$, then $|S_k^1 \cap S_k^2| = 2k-1$, $S_k^1 \setminus S_k^2 = \{a_k\}$ and $S_k^2 \setminus S_k^1 = \{b_k\}$. Thus, the claim holds.
\end{proof}

If $n = 2\ell$, then every schedule $s$ produced by Algorithm~\ref{alg:AlgNext} has exactly $\ell$ layers. Notice that for any $k = 1, \dots, \ell$, the $k^{th}$ layer is $H_k^{(s)} = S_k \setminus S_{k-1}$. Specifically, the last layer is $H_\ell^{(s)} = S_\ell \setminus S_{\ell-1}$. According to Claim~\ref{cl:two_options}, there are at most $2$ options for $S_{\ell-1}$, and since $S_\ell = J$, there are at most $2$ options for the last layer $H_\ell^{(s)}$.

Based on the above observations, we present Algorithm~\ref{alg:2MachinesUnitJobs}.

\begin{algorithm}[H]
\caption{- Determines existence of a NE in $G \in \G^{P2}_{unit}$ and produces one if it exists} 
\label{alg:2MachinesUnitJobs}
\begin{algorithmic}[1]
\IF{$n$ is odd}
\STATE Use Algorithm~\ref{alg:AlgNext} to produce a schedule $s$ and assign the first $n - 1$ jobs.
\STATE Assign the last job in $s$ according to Claim~\ref{cl:odd_n_NE}.
\STATE Return $s$
\ENDIF
\STATE Determine all different options for $S_\ell, S_{\ell-1}$, where $n = 2\ell$.
\FOR{each option of $S_{\ell-1}$}
\STATE Compute $H_\ell^{(s)} = S_\ell \setminus S_{\ell-1}$, where $s$ is a schedule corresponding to $S_{\ell-1}$.
\IF{$H_\ell^{(s)}$ satisfies the condition in Claim~\ref{cl:stable_schedule}}
\STATE Return $s$
\ENDIF
\ENDFOR
\STATE Return "no NE exists"
\end{algorithmic}
\end{algorithm}

The following claims establishes the statement of Theorem~\ref{thm:2MachinesUnitJobs}.

\begin{claim}
\label{cl:AlgUnitCorrect}
Algorithm~\ref{alg:2MachinesUnitJobs} returns a schedule $s$ iff a NE exists and $s$ is NE.
\end{claim}

\begin{proof}
Let $s$ be the schedule returned by Algorithm~\ref{alg:2MachinesUnitJobs}. If $s$ is returned in Step $4$ of the algorithm, then by Claim~\ref{cl:odd_n_NE} it is a NE. Otherwise, $s$ is returned in Step $10$ of the algorithm, and by Claim~\ref{cl:stable_schedule} it is also a NE.

For the other direction, assume a NE exists. If $n$ is odd, then by Claim~\ref{cl:odd_n_NE}, a NE must exist and it is returned by the algorithm in Step $4$. Otherwise, $n = 2\ell$. Let $s$ be a NE. By Claim~\ref{cl:anyNE}, there is en execution of Algorithm~\ref{alg:AlgNext} that produces a schedule, $s'$, with the same layers as $s$, and in which the jobs in the last layer $H_\ell^{(s)} = H_\ell^{(s')}$ are processed on the same machines as in $s$. The corresponding sets $S_{\ell-1}, S_\ell$ are considered in Step $6$. Since $s$ is a NE, $H_\ell^{(s')}$ satisfies the condition in Claim~\ref{cl:stable_schedule}, therefore $s'$ is a NE and is returned by the algorithm in Step $10$.
\end{proof}

\begin{claim}
\label{cl:AlgUnitRunTime}
Algorithm~\ref{alg:2MachinesUnitJobs} has a runtime of $O(n)$.
\end{claim}

\begin{proof}
Algorithm~\ref{alg:AlgNext} runs in time $O(n\log m)$ which is $O(n)$ for $m=2$. Therefore, the complexity of Algorithm~\ref{alg:2MachinesUnitJobs} depends on the number of options for $S_{\ell-1}$ and $S_\ell$ considered in Step $6$. From Claim~\ref{cl:two_options} there are at most two such options, that are identified in time $O(n)$.
\end{proof}

The statement of Theorem~\ref{thm:2MachinesUnitJobs} follows from Claims~\ref{cl:AlgUnitCorrect} and~\ref{cl:AlgUnitRunTime}.
\end{proof}

\section{Related Machines - Equilibrium Existence and Computation}
\label{sec:related}

In this section, we consider the class $\G_{unit}^{Q2}$, in which unit-jobs are assigned to two related machines. W.l.o.g., the two machines, $M_1$ and $M_2$, have rates $r_1 = 1$ and $r_2 = r \leq 1$. Thus, $M_1$ is denoted {\em the fast machine}, while $M_2$ is denoted {\em the slow machine}.

In the following analysis, we distinguish between rational $r$, that is, $r = \frac{a}{b}$ for some integers $a \leq b$, and irrational $r$. Our analysis is based on analyzing the set of possible outputs of Algorithm~\ref{alg:AlgNextRelated}, presented in Section~\ref{sec:alg2}. Recall that the algorithm assigns the jobs greedily according to their completion time.
The following observations are valid for both global and machine-dependent priority lists.

\begin{theorem}
If $r$ is an irrational number, then $G$ has a NE, and it can be calculated in linear time.
\end{theorem}

\begin{proof}
Let $s$ be a schedule produced by Algorithm~\ref{alg:AlgNextRelated}. Since $r$ is irrational, there are no two jobs with the same completion time in $s$. Thus, by Claim~\ref{cl:beneficial} no job has a beneficial migration, and $s$ is a NE.
\end{proof}

We turn to consider the case that $r$ is a rational number, that is, $r = \frac{a}{b}$ for some integers $a \leq b$.
Let $s$ be a schedule produced by Algorithm~\ref{alg:AlgNextRelated}. Let $n = \ell \cdot (a+b) + c$ for $0 \leq c < a + b$. 
The algorithm assigns the jobs in blocks of $a+b$ jobs. In every block, there are $b$ jobs on the fast machine and $a$ jobs on the slow machine (see Figure~\ref{fig:blocks}). Specifically, for $k \geq 1$, the $k^{th}$ block is $B_k(s) = \{j ~|~ (k - 1) \cdot b < C_j(s) \leq k \cdot b \}$. When $s$ is clear from the context, we omit it. We denote the jobs in $B_k$ by non-decreasing order of their completion time $j_1^{(k)}, j_2^{(k)}, \dots, j_{a+b}^{(k)}$.

Note that the last two jobs in a block, $(j_{a+b-1}^{(k)}, j_{a+b}^{(k)})$, have the same completion time, and that these are the only pairs of jobs that share a completion time.

\begin{figure}[h]
  \centering
  \includegraphics[width=1\textwidth]{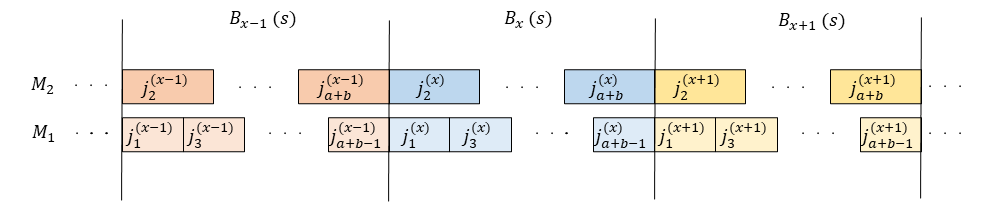}
  \caption{Blocks in schedule $s$ on two related machines having rates $1$ and $r=a/b$.}
  \label{fig:blocks}
\end{figure}

In the sequel we show that for games with a global priority list, having $c \neq 0$ is a necessary and sufficient condition for having a NE, while for games with machine-dependent priority lists, games for which $c=0$ may also have a NE assignment, and their analysis is more involved.



We start by showing that, without any assumptions regarding the priority lists, having $c \ne 0$ is a sufficient condition to ensure that a NE exists and can be produced efficiently. Recall that Algorithm~\ref{alg:AlgNextRelated} produces a schedule stable against cost-reducing deviations. We utilize this property to prove the following claim.

\begin{claim}
\label{cl:c=!0}
For every $G \in \G_{unit}^{Q2}$, if $c \neq 0$, then $G$ has a NE, and a NE can be computed in time $O(n)$.
\end{claim}

\begin{proof}
Let $n = \ell \cdot (a+b) + c$, with $c \neq 0$. If $c \geq \ceil{\frac{b}{a}}$, then let $s$ be a schedule produced by Algorithm~\ref{alg:AlgNextRelated}, and let $j_1$ and $j_2$ be the last jobs processed on $M_1$ and $M_2$, respectively. Because $c \geq \ceil{\frac{b}{a}}$, both $j_1$ and $j_2$ are in $B_{\ell+1}$, and neither of them has a job on the opposite machine with the same completion time. Thus, by Theorem~\ref{thm:NE_related}, $s$ is a NE.

If $0 < c < \ceil{\frac{b}{a}}$, assume that Algorithm~\ref{alg:AlgNextRelated} is executed. After the first $\ell \cdot (a+b) - 2$ jobs are assigned, we enter the $({\ell \cdot (a+b) - 1})^{th}$ iteration. Let $j_1$ be the first unassigned job in $\pi_1$. In Step $3$ of the algorithm, machine $i^*$ might be either $M_1$ or $M_2$. Suppose the tiebreaker resolves in favor of assigning $j_1$ to $M_1$. Let $s$ be the produced schedule, and let $j_2$ be the last job processed on $M_2$. Since $0 < c < \ceil{\frac{b}{a}}$, it follows that $L_1(s) > L_2(s) / r$. Note that $C_{j_1}(s) = C_{j_2}(s)$ and $j_1 \prec_1 j_2$. Therefore, by Theorem~\ref{thm:NE_related} $s$ is a NE.
\end{proof}

\subsection{Global Priority List}

Assume that both machines have the same priority list $\pi = \langle 1,2,\ldots,n \rangle$.

\begin{theorem}
Let $G$ be a game with $m = 2$ related machines, $r_1 = 1$, $r_2 = r = \frac{a}{b}$ for two integers $a \leq b$, a global priority list $\pi$ and $n = \ell \cdot (a+b) + c$ unit-jobs for $0 \leq c < a + b$. $G$ has a NE iff $c \neq 0$.
\end{theorem}

\begin{proof}
If $c \ne 0$, then by Claim~\ref{cl:c=!0} a NE exists and can be produced in linear time.

For the other direction, let $s$ be a NE schedule, and let $j_1$ and $j_2$ denote the last jobs processed on $M_1$ and $M_2$, respectively. Assume by contradiction that $c=0$. Schedule $s$ must be balanced, since otherwise, the last job can benefit from migration. Therefore, $L_1(s) = L_2(s) / r$, and by Theorem~\ref{thm:NE_related}, since $\pi$ is global, $s$ is not a NE.
\end{proof}

\subsection{Machine-Dependent Priority Lists}

We turn to discuss instances in which $\pi_1 \neq \pi_2$. We present a linear time algorithm for deciding whether a given game has a NE, and producing a NE if one exists.

As a warm-up, we demonstrate that a tie-breaking  decision during a run of Algorithm~\ref{alg:AlgNextRelated} may by crucial for the stability of the resulting schedule. This example highlights the need to trace multiple possible outcomes of the algorithm.

\begin{example}
{\em
Consider the game $G$ with $m = 2$ related machines, where $r = \frac{2}{3}$ and $J = \{1, 2, \dots, 10\}$. The priority lists are:
\[\pi_2 = \langle 1, 2, 3, 4, 10, 9, 6, 7, 8, 5 \rangle\]
\[\pi_1 = \langle 1, 2, 3, 4, 5, 6, 7, 8, 9, 10 \rangle\]
Now, consider the scheduling process of Algorithm~\ref{alg:AlgNextRelated}. After assigning the first three jobs, a tie-break occurs as job $4$ is the first unassigned job in both $\pi_1$ and $\pi_2$, with its possible completion time on both machines being equal. At this point, the algorithm's outcome may diverge, as job $4$ can either be assigned to $M_1$ or $M_2$.
If the algorithm assigns job $4$ to $M_2$, then the resulting schedule is a NE (see Figure~\ref{fig:tie_break_1}), as the last jobs to complete, jobs $8$ and $9$, are each prioritized on their respective machines over the other. However, if job $4$ is assigned to $M_1$, the last jobs to complete are jobs $7$ and $8$ (see Figure~\ref{fig:tie_break_2} ($b_1$) and ($b_2$)). Since job $7$ is prioritized over job $8$ on both priority lists, an endless sequence of best response moves arises between them, preventing convergence to a NE.
}

\begin{figure}[h]
  \centering
  \includegraphics[width=1\textwidth]{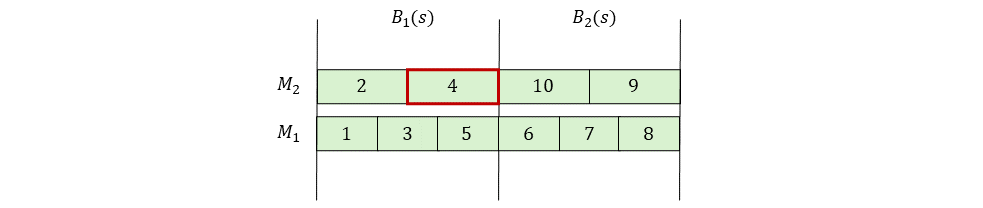}
  \caption{The NE schedule produced by Algorithm~\ref{alg:AlgNextRelated} in case that job $4$ is assigned to $M_2$. }
  \label{fig:tie_break_1}
\end{figure}

\begin{figure}[h]
  \centering
  \includegraphics[width=1\textwidth]{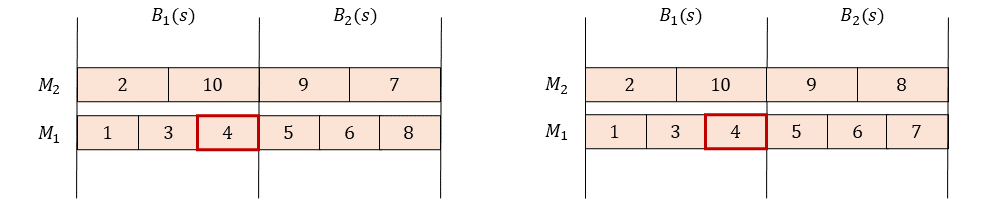}
  \caption{The unstable schedules Algorithm~\ref{alg:AlgNextRelated} produces when job $4$ is assigned to $M_1$. In both schedules, job $7$ can reduce its rank from $9.5$ to $9$ by deviating to the other machine and precede job $8$.}
  \label{fig:tie_break_2}
\end{figure}


\end{example}

By manipulating the priority lists a bit, we can build an instance in which an assignment of job $4$ on the fast machine yields a NE, while its assignment on the slow machine leads to an unstable schedule. Thus, no tie-breaking rule can be used to promote stability. On the other hand, as we show, even if the algorithm faces several tie-breaking decisions during its run, there are only two possible sets of jobs that may be processed last in any schedule produced by the algorithm.

\begin{theorem}
\label{thm:2MachinesUnitJobsRelated}
For every $G \in \G_{unit}^{Q2}$ with $r_1 = 1$ and $r_2 = r = \frac{a}{b}$ for some $a \leq b$, it is possible to decide in linear time whether $G$ has a NE, and to produce a NE if one exists.
\end{theorem}

\begin{proof}
Let $n = \ell \cdot (a+b) + c$. If $c \ne 0$, then by Claim~\ref{cl:c=!0} a NE exists and can be produced in linear time.
However, if $c = 0$, then the existence of a NE is not guaranteed. We begin by showing that if a NE does exist, then Algorithm~\ref{alg:AlgNextRelated} has the capability to produce a schedule with the same blocks.

\begin{claim}
\label{cl:anyNERelated}
Any NE assignment $s$ has an execution of Algorithm~\ref{alg:AlgNextRelated} that produces a schedule, $s^*$, with the same blocks as in $s$, and in which the jobs in the last block are processed on the same machines as in $s$.
\end{claim}

\begin{proof}
%
%
%
Let $s$ be a NE assignment. The proof is by induction on the number of blocks. We show that for any $1 \leq k \leq \ell$, there exists an execution of Algorithm~\ref{alg:AlgNextRelated} that produces a schedule, $s^*$, with the same first $k$ blocks as in $s$. The base case refers to the first block.
Since $s$ is a NE, for all $j,t \in J$, if $C_j(s) < C_t(s)$ and $s(j) = M_i$ then $j \prec_i t$.
Let the jobs in $B_1(s)$ be $\langle j_1, j_2, \dots j_{a+b-1}, j_{a+b} \rangle$ ordered in a non-decreasing order of their completion time. Note that $j_{a+b-1}$ and $j_{a+b}$ have the same completion time. Assume, w.l.o.g., that $s(j_{a+b-1}) = M_1$ and $s(j_{a+b}) = M_2$. Among $j_1, \dots, j_{a+b-2}$, no two jobs have the same completion time. Therefore, any execution of Algorithm~\ref{alg:AlgNextRelated} produces a schedule in which these jobs are in $B_1$, on the same machines as in $s$.

Now consider $j_{a+b-1}$ and $j_{a+b}$, and consider the scheduling process of Algorithm~\ref{alg:AlgNextRelated}. Recall that $s(j_{a+b-1}) = M_1$ and $s(j_{a+b}) = M_2$. In $\pi_1$, among the jobs $j_{a+b-1}$, $j_{a+b}$ and all the jobs in blocks $2, \dots, \ell$, either $j_{a+b-1}$ is first or $j_{a+b-1}$ is second and $j_{a+b}$ is first. Similarly, in $\pi_2$, among the jobs $j_{a+b-1}$, $j_{a+b}$ and all the jobs in blocks $2, \dots, \ell$, either $j_{a+b}$ is first or $j_{a+b}$ is second and $j_{a+b-1}$ is first. In any of these options, Algorithm~\ref{alg:AlgNextRelated} has an execution that creates a schedule $s^*$ in which $B_1(s^*) = B_1(s)$. Note that jobs $j_{a+b}$ and $j_{a+b-1}$ might be assigned to the opposite machine with respect to their assignment in $s$.

For the induction step, let $2 \leq k \leq \ell$. By the induction hypothesis, there exists a schedule $s^*$ produced by Algorithm~\ref{alg:AlgNextRelated}, that agrees with $s$ on the first $k-1$ blocks. Let the jobs in $B_k(s)$ be $\langle j_1, j_2, \dots j_{a+b-1}, j_{a+b} \rangle$ ordered in a non-decreasing order of their completion time. The arguments applied above for $B_1(s)$ can now be applied on $B_k(s)$ to conclude that $s^*$ can be extended such that its $k$-th block consists of the jobs in $B_k(s)$.

Lastly, consider the last block $B_\ell(s)$, where $j_{a+b-1}$ and $j_{a+b}$ are the last jobs processed in the schedule. Since $s$ is a NE, it follows that $j_{a+b-1} \prec_1 j_{a+b}$ and $j_{a+b} \prec_2 j_{a+b-1}$. 
Therefore, independent of whether the algorithm assigns $j_{a+b}$ before or after $j_{a+b-1}$, in every execution of the algorithm in which the jobs in the last block are the same jobs as in $B_\ell(s)$, jobs $j_{a+b-1}$ and $j_{a+b}$ are assigned to $M_1$ and $M_2$, respectively, as in $s$.
\end{proof}

Combining this claim with Claim~\ref{cl:beneficial}, we infer that by examining all the potential last blocks in a schedule produced by Algorithm~\ref{alg:AlgNextRelated}, we can determine the existence of a NE and even produce one if it exists. We now show that at most two potential last blocks exist, and we can identify them efficiently in linear time.

Recall that $n = \ell \cdot (a+b)$. For $1 \leq k \leq \ell$, let $\Gamma_k$ be the set of sets of jobs such that $S_k \in \Gamma_k$ if and only if $|S_k|= k \cdot (a+b)$ and there exists a run of Algorithm~\ref{alg:AlgNextRelated} in which the jobs of $S_k$ are assigned on the first $k$ blocks.
We show that for $1 \leq k \leq \ell$, $|\Gamma_k| \leq 2$. Moreover, if $|\Gamma_k|=2$ then the two sets are identical up to a single job.

For $1 \leq t \leq n$, denote by $\hat{\pi_i}^{(t)}$ the unassigned jobs in $\pi_i$, at the beginning of the $t^{th}$ iteration of Algorithm~\ref{alg:AlgNextRelated}, i.e., after assigning the first $t - 1$ jobs, and before assigning the $t^{th}$ job. Note that $\hat{\pi_i}^{(1)} = \pi_i$, and that $|\hat{\pi_i}^{(t)}| = n - t + 1$. Denote by $\hat{\pi_i}^{(t)}(x)$ the job in the $x^{th}$ place in $\hat{\pi_i}^{(t)}$.
Let $t_k = k \cdot (a + b) - 1$. Note that at the beginning of the $t_k^{th}$ iteration of the algorithm, $a + b - 2$ jobs have already been assigned within block $k$, with the last two jobs yet to be assigned.



\begin{claim}
\label{cl:two_options_Related}
For all $1 \leq k \leq \ell$, $|\Gamma_k| \le 2$. If $S_k^1, S_k^2 \in \Gamma_k$, then
\begin{enumerate}
    \item $|S_k^1 \cap S_k^2| = k \cdot (a + b) - 1$, and
    \item $\hat{\pi_1}^{(t_k)}(1) = \hat{\pi_2}^{(t_k)}(1)$, $\hat{\pi_1}^{(t_k)}(2) \neq \hat{\pi_2}^{(t_k)}(2)$, and 
    \item denote $\hat{\pi_1}^{(t_k)}(2) = a_k$ and $\hat{\pi_2}^{(t_k)}(2) = b_k$, then $S_k^1 \setminus S_k^2 = \{a_k\}$, $S_k^2 \setminus S_k^1 = \{b_k\}$, w.l.o.g.
\end{enumerate}
\end{claim}

\begin{proof}
The proof is by induction on $k$.
Note that for every block, the assignment of the first $a+b-2$ jobs in the block is deterministic, and a tie-breaking may be required only for the assignment of the $(a+b-1)$-st job. In other words, different blocks may be created when assigning the last two jobs in a block, depending on the tie breaking applied for the second to last job in the block. Recall the analysis for two identical machines, given in the proof of Claim~\ref{cl:two_options}. This analysis fits the case $a=b=1$. Indeed, for $a=b=1$ every block consists on two jobs, and a run of the algorithm may split if  $\hat{\pi_1}^{(t_k)}(1) = \hat{\pi_2}^{(t_k)}(1)$, and $\hat{\pi_1}^{(t_k)}(2) \neq \hat{\pi_2}^{(t_k)}(2)$.

Note that, independent of the machines rate, if $S_k^1, S_k^2 \in \Gamma_k$ and $S_k^1 \setminus S_k^2 = \{a_k\}$, $S_k^2 \setminus S_k^1 = \{b_k\}$, then, when assigning block $k+1$ starting from $S_k^1$, job $b_k$ has the highest priority in its list and is guaranteed to be included in $B_{k+1}$, and when assigning block $k+1$ starting from $S_k^2$, job $a_k$ has the highest priority in its list and is guaranteed to be included in $B_{k+1}$. Based on this observation, and the fact that $a+b-2$ of the jobs in each block are assigned deterministically, the proof of Claim~\ref{cl:two_options} can be generalized to the case of arbitrary $a$ and $b$. We omit the technical details.
\end{proof}

The following algorithm decides whether a given game has a NE profile.


\begin{algorithm}[H]
\caption{- Algorithm to determine existence of a NE in $G$ and produce one if it exists} \label{alg:2MachinesUnitJobsRelated}
\begin{algorithmic}[1]
\IF{$c \neq 0$}
\STATE Use Algorithm~\ref{alg:AlgNextRelated} to produce a schedule $s$
\IF{$s$ is not a NE}
\STATE Adjust $s$ to resolve the tiebreaker according to Claim~\ref{cl:c=!0}
\ENDIF
\STATE Return $s$
\ENDIF
\STATE Compute $\Gamma_\ell, \Gamma_{\ell-1}$, for $\ell = n / (a+b)$
\FOR{each $S_{\ell-1} \in \Gamma_{\ell-1}$}
\STATE Compute $B_\ell(s) = S_\ell \setminus S_{\ell-1}$, where $s$ is a schedule related to $S_{\ell-1}$
\IF{$s$ is a NE}
\STATE Return $s$
\ENDIF
\ENDFOR
\STATE Return "no NE exists"
\end{algorithmic}
\end{algorithm}

\begin{claim}
\label{cl:AlgUnitCorrectRelated}
Algorithm~\ref{alg:2MachinesUnitJobsRelated} returns a schedule $s$ iff a NE exists and $s$ is NE.
\end{claim}

\begin{proof}
Let $s$ be the schedule returned by Algorithm~\ref{alg:2MachinesUnitJobsRelated}. If $s$ is returned in Step $6$ of the algorithm, then by Claim~\ref{cl:c=!0} it is a NE. Otherwise, $s$ is returned in Step $12$ and it is a NE.

For the other direction, assume a NE exists, and $n = \ell \cdot (a+b) + c$. If $c \neq 0$, then by Claim~\ref{cl:c=!0}, a NE must exist and it is returned by the algorithm in Step $6$. Otherwise, let $s$ be a NE. By Claim~\ref{cl:anyNERelated}, there is an execution of Algorithm~\ref{alg:AlgNextRelated} that produces a schedule, $s'$, with the same blocks as in $s$, and in which the jobs in the last block $B_\ell(s) = B_\ell(s')$ are processed on the same machines as in $s$. The corresponding sets $S_{\ell-1}, S_\ell$ will be considered in Steps $9$ and $10$. Since $s$ is a NE, $s'$ is also a NE and will be returned by the algorithm in Step $12$.
\end{proof}

\begin{claim}
\label{cl:AlgUnitRunTimeRelated}
Algorithm~\ref{alg:2MachinesUnitJobsRelated} has a runtime of $O(n)$.
\end{claim}

\begin{proof}
First note that Algorithm~\ref{alg:AlgNextRelated} runs in time $O(n\log m)$ which is $O(n)$ for $m=2$. Also note that verifying whether a schedule $s$ is a NE in Steps $3$ and $11$ can be done in $O(n)$ time by checking if every job's current assignment is its best-response. Therefore, the complexity of Algorithm~\ref{alg:2MachinesUnitJobsRelated} depends on the number of sets in $\Gamma_{\ell-1}$ and $\Gamma_\ell$ considered in Steps $9$ and $10$. From Claim~\ref{cl:two_options_Related} there are at most two such options, that are identified in time $O(n)$.
\end{proof}

The statement of Theorem~\ref{thm:2MachinesUnitJobsRelated} follows from Claims~\ref{cl:AlgUnitCorrectRelated} and~\ref{cl:AlgUnitRunTimeRelated}.
\end{proof}


\section{The Effect of Competition on the Equilibrium Inefficiency}
\label{sec:inefficiency}

In this section, we analyze the equilibrium inefficiency with respect to the objective of minimizing the makespan. The main question we consider is whether competition improves or worsen the equilibrium inefficiency.

We show that the known upper and lower bounds for the PoA given in~\cite{RST21} for games without competition are valid also for games with competition that have a NE. In addition, we show that competition might be beneficial or harmful. Specifically, for the class $\G^P$ we present two games such that for the first game $PoA(G_1)=1$ if the game is played without competition, while with competition $PoS(G_1)= PoA(\G^P)$, and a game for which with competition $PoA(G_2)=1$,  while without competition $PoS(G_2)= PoA(\G^P)$. Similar, though not tight, results are shown for the class $\G^{Q2}$. For games that do not have a NE, we analyze the price of sinking. 


\subsection{Identical Machines}

Recall that Algorithm~\ref{alg:AlgNext} produces a Nash equilibrium for any game on two identical machines with {\em Inversed-Policies}. We first bound the PoA of the resulting NE. Our bound is different from the bound presented in~\cite{CKN04} for inversed policies without competition. They showed that the PoA is at least $\frac{4}{3}$, while we show a tight bound of $\frac{3}{2}$.

\begin{theorem}
\label{Thm:InvPol_tmp}
A coordination mechanism with {\em Inversed-Policies} has a price of anarchy of $\frac{3}{2}$.
\end{theorem}

\begin{proof}
In Theorem~\ref{thm:POAidentical} we present an upper bound of $2 - \frac{1}{m}$ for the PoA of any game played on identical machines. In particular, for any game $G$ on two machines and {\em Inversed-Policies}, PoA$(G) \le 2 - \frac{1}{2} = \frac 3 2$.

For the lower bound, consider $2k$ unit-size jobs $p_1 = p_2 = ... = p_{2k-1} = p_{2k} = 1$, and one job $j^*$ with $p_{j^*} = 2k$.
In the optimal allocation, we have all $2k$ unit-size jobs processed on a single machine, while $j^*$ is processed on a the other machine, resulting in a makespan of $2k$.
However, if $\pi_1 = \langle 1, 2, \dots, k-1, k, j^*, k+1, k+2, \dots, 2k \rangle$, in any Nash equilibrium of this instance, we get unit-jobs $1, \dots, k$ processed on $M_1$ and unit-jobs $k + 1, \dots, 2k$ processed on $M_2$. Job $j^*$ is processed on either $M_1$ or $M_2$, resulting in a makespan of $3k$ and a price of anarchy of $3/2$.
\end{proof}

Next, we analyze the class $\G^P$ of games played on $m$ identical machines, and $n$ arbitrary jobs.

\begin{theorem}
\label{thm:POAidentical}
PoA$(\G^P) = $PoS$(\G^P) = 2 - \frac{1}{m}$.
\end{theorem}

\begin{proof}
Let $G \in \G^P$. Let $W=\sum_j p_j$ be the total processing time of all jobs. Assume that $G$ has a NE, and let $s$ be a NE with maximal makespan. Let $a$ be a last job to complete in $s$, and let $L_1, \dots, L_m$ be the total processing time of all jobs different from $a$ on $M_1, \dots, M_m$. Let $M_k$ be the machine on which job $a$ is processed in $s$. Thus, $C_{max}(s) = L_k + p_a$. Since $s$ is a NE, and job $a$ has highest rank in $s$, no migration of $a$ reduces its completion time. Therefore, for any $1 \le i \le m$, $i \neq k$, $L_k + p_a \leq L_i + p_a$, otherwise, job $a$ could benefit from migration to $M_i$. Therefore, for any $1 \le i \le m$, $C_{max}(s)\leq L_i + p_a$. Also, $OPT(G) \ge \frac{W}{m}$ and $OPT(G) \ge p_a$. Combining these inequalities yields
\[C_{max}(s)\leq \frac{W + (m-1) \cdot p_a}{m} \leq (2 - \frac{1}{m}) \cdot OPT(G),\]
implying that $PoA(G) \leq 2 - \frac{1}{m}$.

For the lower bound, we show that the classical tight analysis of List-scheduling can be adapted to our game by choosing appropriate priority lists. Given $m > 1$, the following is an instance for which PoS$=2 - \frac 1 m$.
The set of jobs includes a single job $j^*$ with processing time $m$ and $m(m-1)$ unit jobs $j_1, \dots, j_{m(m-1)}$. Let $J_{last}$ denote the set consisting of the last $m$ unit-jobs. For $1 \le i \le m$, let $j_{last}^i$ be the $i$-th job in $J_{last}$. Note that $j_{last}^i$ is the $m(m-2)+i$-th unit job in the instance. 
The priority list of machine $M_i$, $\pi_i$, starts with the $m(m-1)$ unit jobs, which appear according to their index, except for $j_{last}^i$, which is moved to be last among the unit-jobs. Finally, the long job $j^*$ is last in $\pi_i$.

It is easy to verify that in every NE profile, for every $1 \le k \le m-1$, the $m$ unit-jobs $m(k-1)+1,\ldots, mk-1$ are processed in slot $k$. Also, by Claim~\ref{cl:beneficial}, Job $j_{last}^i$ is processed on machine $i$ in slot $m-1$. Job $j^*$ is assigned as last on any machine. Thus, in every NE profile, some machine has load $2m-1$. On the other hand, an optimal assignment assigns the heavy job on a dedicated machine, and partition the unit-jobs in a balanced way among the remaining $m - 1$ machines. In this profile, all the machines have load $m$. The corresponding PoS is $\frac{2m-1}{m} = 2 - \frac 1 m$.
\end{proof}

\noindent{\bf The Benefit or Loss due to Competition:} We show that competition may be both helpful and harmful for the social cost.
We first present an instance $G_1 \in \G^P$ for which the corresponding game without competition has $PoS(G_1)= 2 - \frac 1 m$, while the same game with competition has $PoA(G_1)=1$.

\begin{example}
{\em
Consider the game $G_1$ with $m$ identical machines and $n = 2m$ jobs. The set of jobs consists of $m - 2$ jobs $X = \{x_3, x_4, \dots, x_m$\}, such that for $x_i \in X,\ p_{x_i} = m - 1$, additional $m - 2$ jobs $Y = \{y_3, y_4, \dots, y_m$\}, such that for $y_i \in Y,\ p_{y_i} = 1$, and $4$ special jobs $a, b, c, d$, where $p_a = m - 1,\ p_b = 1 - \epsilon,\ p_c = \epsilon,\ p_d = m$. The priority lists are $\pi_1 = \langle b, a, c, Y, d, X \rangle$, $\pi_2 = \langle a, d, c, Y, b, X \rangle$, and for $3 \le i \le m,\ \pi_i = \langle x_i, X \setminus \{x_i\}, y_i, Y \setminus \{y_i\}, c, d, a, b \rangle$. When a list includes a set, it means that the set elements appear in arbitrary order. We start by analyzing the NE schedules without competition. In any NE (see Figure~\ref{fig:example_1.2}), each job $x_i$ is processed alone on $M_i$, job $a$ is processed on $M_2$, jobs $b$, $c$ and all jobs of $Y$ are processed on $M_1$. Finally, job $d$ is processed on any of the machines. The makespan of any NE is therefore $2m - 1$. However, with rank-based utilities, these schedules are not stable, since job $a$ can benefit from a migration to $M_1$. Such a migration delays its completion time, and improves its rank from $\frac{3m-1}{2}$ to $m$. Consequently, migrations of job $d$ to $M_2$, and each job $y_i$ to $M_i$, are beneficial, and result in an optimal schedule, which is the unique NE on the game. The makespan of this schedule is $m$. Thus, for the game $G_1$, the PoS without competition is $\frac{2m -1}{m} = 2 - \frac 1 m$, while with competition PoA($G_1$) $= 1$.
}

\begin{figure}[h]
  \centering
  \includegraphics[width=1\textwidth]{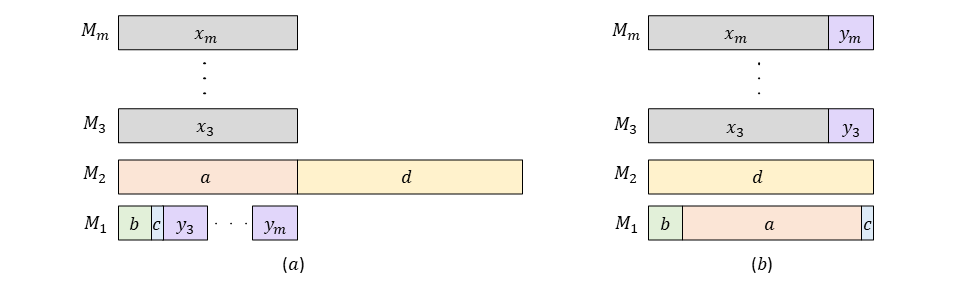}
  \caption{The game $G_1$. (a) NE without competition, (b) NE with competition}
  \label{fig:example_1.2}
\end{figure}

\end{example}

On the other hand, we now show that competition might also increase the equilibrium inefficiency. Specifically, in the following example, we describe a game $G_2 \in \G^P$ in which without competition $PoA(G_2)=1$, and with competition $PoS(G_2)=2- \frac 1 m$.

\begin{example}
{\em
Consider the game $G_2$ with $m$ identical machines and $n = 4m - 3$ jobs. The set of jobs consists of $2 \cdot (m - 2)$ jobs $X = \{x_{i1}, x_{i2}\ |\ 3 \le i \le m\}$, such that for $x_{i1}, x_{i2} \in X,\ p_{x_{i1}} = m - i + 1 + \epsilon,\ p_{x_{i2}} = i - 2$, additional $m - 2$ jobs $Y = \{y_i\ |\ 3 \le i \le m\}$, such that for $y_i \in Y,\ p_{y_i} = 1$, additional $m$ jobs $C = \{c_1, c_2, \dots, c_m\}$, such that for $c_i \in C,\ p_{c_i} = \epsilon$, and $3$ special jobs $a, d, e$, where $p_a = m - 1 - m \cdot \epsilon,\ p_d = 1,\ p_e = m$. The priority lists are $\pi_1 = \langle c_1, a, c_2, \dots, c_m, e, d, X, Y \rangle$, $\pi_2 = \langle a, d, y_m, y_{m-1}, \dots, y_3, e, C, X \rangle$, and for $3 \le i \le m,\ \pi_i = \langle x_{i1}, y_i, x_{i2}, a, d, e, C, X \setminus \{x_{i1}, x_{i2}\}, Y \setminus \{y_i\} \rangle$. When a list includes a set, it means that the set elements appear in arbitrary order. In the unique no-competition NE (see Figure~\ref{fig:example_2.4}), all the jobs in $C$ and job $e$ are processed on $M_1$, jobs $a$ and $d$ are processed on $M_2$, and for $3 \le i \le m$, jobs $x_{i1},\ y_i$ and $x_{i2}$ are processed on $M_i$. The makespan is $m(1 + \epsilon)$. However, with rank-based utilities, this schedule is not stable, since job $a$ can benefit from a migration to $M_1$. Such a migration delays its completion time, but improves its rank from $3m - 4$ to $2m - 2$. Consequently, migrations of jobs $y_m,\ y_{m-1},\ \dots,\ y_3$ to $M_2$ are beneficial, and result in a unique NE, with makespan $2m - 1$. Thus, for the game $G_2$, the PoS with competition is $\frac{2m -1}{m(1 + \epsilon)} \to 2 - \frac 1 m$, while without competition PoA($G_2$) $\to 1$.
}


\begin{figure}[h]
  \centering
  \includegraphics[width=1\textwidth]{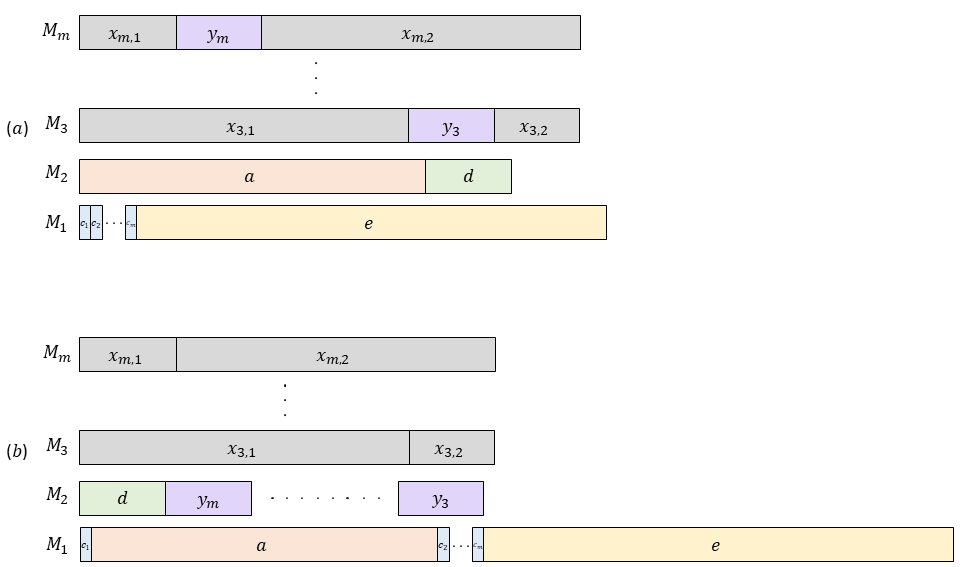}
  \caption{The game $G_2$. (a) NE without competition, (b) NE with competition}
  \label{fig:example_2.4}
\end{figure}

\end{example}

\subsubsection{Sink Equilibria Analysis - Identical Machines}

Let $G \in \G^{P,global}_{unit}$ be a game played by $n = \ell \cdot m + c$, for $0 \leq c < m$, unit-length jobs, on $m$ identical machines, with a global priority list $\pi$. Assume that $G$ has no NE. By Theorem~\ref{thm:Unit_m2}, either $m \ge 3$ or $m = 2$ and $n$ is even. We analyze the sink equilibria of $G$. As we show, natural dynamics lead to sink equilibria with good social cost.

Assume w.l.o.g., that $J = \{1, \dots, n\}$ and that $\pi = \langle 1, \dots, n \rangle$. Given a profile $s$, let $Sub(s)$ be the set of suboptimal players in $s$, and let $Lag(s)$ denote the set of jobs such that $j \in Lag(s)$ iff $rank_j(s) > m \cdot \ceil{j/m} - \frac{m-1}{2}$. Note that $Lag(s) \subseteq Sub(s)$.

Recall that BRD is a natural dynamics, in which, as long as the system is not in a stable state, a suboptimal player is chosen by applying some deviator rule, and performs a best response move. We suggest the following deviator rule for selecting the next player to deviate in an (infinite) sequence of best response moves, and bound the expected cost of a profile in the resulting sink equilibrium.

\paragraph{Priority-Based Deviator Rule:}
If $Lag(s) \neq \emptyset$, then choose the player with the highest priority in $Lag(s)$. Else, if $Sub(s) \setminus Lag(s) \neq \emptyset$ then choose the player with the lowest priority in $Sub(s)$.

In the social optimum, the load is balanced. Thus, for every $G \in \G^{P,global}_{unit}$, $SO(G) = \ceil{\frac n m}$. Specifically, if $c = 0$ then $SO(G) = \ell$, and if $c > 0$ then $SO(G) = \ell + 1$. We start by analyzing instances with $c \ne 1$. 

\begin{theorem}
\label{thm:sink_global_c=0}
Let $G \in \G^{P,global}_{unit}$ be a game with $n = \ell \cdot m + c$ jobs. If $G$ has no NE, and $c \ne 1$, then $PoSINK(G,priority-based) = 1 + \frac{1}{2 \cdot \ceil{\frac n m}}$.
\end{theorem}

\begin{proof}
Let $G \in \G^{P,global}_{unit}$, let $D$ denote the priority-based deviator rule, and let $\mathcal{Q}(G, D)$ be the set of sink equilibria of $G$ reached by applying BRD with $D$. Assume $G$ has no pure NE. We show that if $c \ne 1$, then for every sink equilibrium $Q \in \mathcal{Q}(G, D)$, $SC(Q) = \ceil{\frac n m} + \frac 1 2$.

Consider a BR sequence applied with the priority-based deviator rule. Let $s_1$ be an initial schedule of such a sequence. If $Lag(s_1) \neq \emptyset$, let $j$ be the job with the highest priority in $Lag(s_1)$. Let $k = \ceil{j/m}$. There are less than $m \cdot k$ jobs that precede $j$ in $\pi$. Therefore, there exists a machine with at most $k - 1$ such jobs. By migrating to this machine, yielding the schedule $s_2$, job $j$ will have completion time at most $k$, and will share rank at most $m \cdot k - \frac{m-1}{2}$ with other jobs. By definition of $Lag(s_1)$, $rank_j(s_1) > m \cdot k - \frac{m-1}{2}$, therefore this is a beneficial migration. Also, $j \notin Lag(s_2)$, and since $j$ has highest priority in $\pi$ among the jobs in $Lag(s_1)$, $j$ cannot re-enter $Lag(s_i)$ for any later schedule $s_i$ in the BR sequence. Therefore, by repeatedly applying the priority-based deviator rule, the BR sequence reaches a schedule $s_t$, in which $Lag(s_t) = \emptyset$. Note that $s_t$ is a balanced schedule, in which every job $j$ has completion time $k = \ceil{j / m}$.

Since $c \ne 1$, jobs $n$ and $n - 1$ are among the last jobs to complete in $s_t$, both have completion time $\ell$ (if $c = 0$) or $\ell + 1$ (if $c \ge 2)$. Thus, according to the deviator rule, the next job to act is the lowest-priority job in $Sub(s_t)$, i.e., job $n - 1$. The only beneficial deviation of job $n - 1$ is into the machine on which job $n$ is processed. By bypassing job $n$, job $n - 1$ maintains its completion time, and improves its rank by $\frac 1 2$. In the updated schedule $s_{t+1}$, job $n$ now has completion time $\ceil{n / m} + 1$ and rank $n$, meaning $n \in Lag(s_{t+1})$. Following the deviator rule, $n$ then migrates to any least-loaded machine. This BR sequence continues alternating between a deviation of player $n - 1$ and player $n$.

Thus, for any sink equilibrium $Q$, we have that $SC(Q)=\frac 1 2 \cdot \ceil{\frac n m} + \frac{1}{2} \cdot (\ceil{\frac n m} + 1) = \ceil{\frac n m} +\frac {1} 2$, implying that $PoSINK(G,D) = 1 + \frac{1}{2 \cdot \ceil{\frac n m}}$
\end{proof}

We now analyze instances with $c = 1$. Since for $m = 2$ and $c = 1$ a NE is guaranteed to exist, we assume $m \ge 3$.

\begin{theorem}
\label{thm:sink_global_c>1}
Let $G \in \G^{P,global}_{unit}$ be a game with $n = \ell \cdot m + c$ jobs. If $G$ has no NE, and $c = 1$, then $PoSINK(G,priority-based) = 1$.
\end{theorem}

\begin{proof}
Let $G \in \G^{P,global}_{unit}$, let $D$ denote the priority-based deviator rule, and let $\mathcal{Q}(G, D)$ be the set of sink equilibria of $G$ reached by applying BRD with $D$. Assume $G$ has no pure NE. We show that if $c = 1$, then for every sink equilibrium $Q \in \mathcal{Q}(G, D)$, $SC(Q) = \ell + 1$.

Consider a BR sequence applied with the priority-based deviator rule. This sequence follows a similar structure to the one analyzed in Theorem~\ref{thm:sink_global_c=0}, up until it reaches a schedule $s_t$, where $Lag(s_t) = \emptyset$. Recall that $s_t$ is a balanced schedule, in which every job $j$ has completion time $k = \ceil{j / m}$. In particular, job $n$ has completion time $\ell + 1$.

Recall that $m \ge 3$. According to the deviator rule, the next job to act is the lowest-priority job in $Sub(s_t)$. Assume first that the job preceding $n$ on its machine, $M_i$, is neither $n - 1$ nor $n - 2$ (see Figure~\ref{fig_sink_c=1_1} (a)). In this case, the deviator rule selects job $n - 2$. The only beneficial migration for job $n - 2$ is to move from its current machine, say $M_k$, to the machine on which job $n - 1$ is processed. By bypassing job $n - 1$, job $n - 2$ maintains its completion time $\ell$, and improves its rank by $\frac 1 2$. In the resulting schedule, $s_{t+1}$, job $n - 1$ has completion time $\ell + 1$, and thus is in $Lag(s_{t+1})$ (see Figure~\ref{fig_sink_c=1_1} (b)). Following the deviator rule, $n - 1$ then migrates to $M_k$, which is the least loaded machine. This sequence continues in a similar manner, where $n - 2$ migrates and bypasses $n - 1$, while $n - 1$, in turn, returns to the machine from which $n - 2$ migrated. The makespan of all the schedules in the sink equilibrium is $\ell + 1$.

\begin{figure}[h]
  \centering
  \includegraphics[width=1\textwidth]{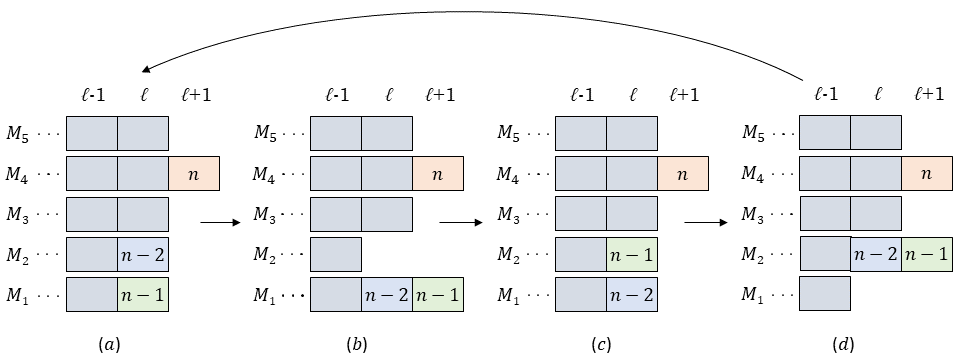}
  \caption{A sink equilibrium where the job preceding $n$ on $M_i$ is neither $n - 1$ nor $n - 2$.}
  \label{fig_sink_c=1_1}
\end{figure}

We now show that this type of sink equilibrium is the only possible outcome. Assume that the job preceding $n$ on $M_i$ is $n - 2$ (see Figure~\ref{fig_sink_c=1_2} (a)). In this case, jobs $n - 1$ and $n - 2$ are not suboptimal, and the deviator rule chooses job $n - 3$, which has two potential BRs, each chosen with probability $\frac 1 2$. It can either migrate from its current machine, say $M_k$, to the machine on which $n - 1$ is processed (option $1$), or it can move to $M_i$, on which $n - 2$ and $n$ are processed (option $2$). If option $1$ is chosen (see Figure~\ref{fig_sink_c=1_2} (d)), then in the next step, job $n - 1$ is in $Lag(s_{t+1})$, and thus being chosen and migrates to $M_k$. In the resulting schedule, job $n - 3$ is selected again and faces the same choices: option $1$ or option $2$. However, if option $2$ is chosen (see Figure~\ref{fig_sink_c=1_2} (b)), the resulting schedule $s_{t+1}$ places jobs $n - 3$, $n - 2$, and $n$ on $M_i$. Here, jobs $n - 2$ and $n$ are in $Lag(s_{t+1})$, and thus job $n - 2$ is chosen and migrates to $M_k$, the least-loaded machine. This leads to a schedule where $n - 3$ is now the job directly preceding $n$ on $M_i$, and the system returns to the sink equilibrium previously described.

\begin{figure}[h]
  \centering
  \includegraphics[width=1\textwidth]{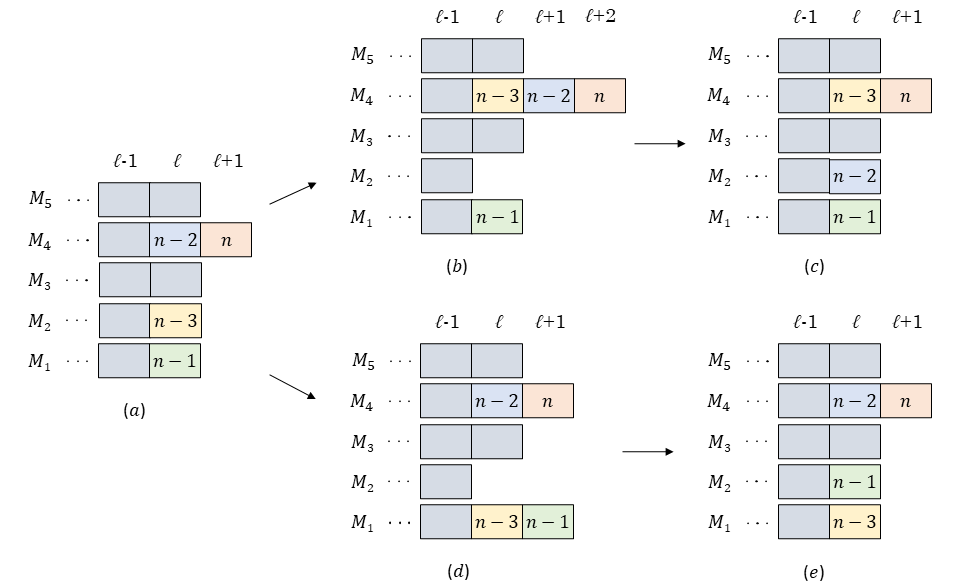}
  \caption{A BRD sequence when the job preceding $n$ on $M_i$ is $n - 2$. (c) The sequence leads to a state where neither $n - 1$ nor $n - 2$ precedes $n$. (e) The system returns to a schedule where $n - 2$ precedes $n$ on $M_i$.}
  \label{fig_sink_c=1_2}
\end{figure}

The final case to consider is when the job preceding $n$ on $M_i$ is $n - 1$ (see Figure~\ref{fig_sink_c=1_3}). In this scenario, the deviator rule selects job $n - 2$, whose only beneficial migration is from its current machine, say $M_k$, to $M_i$, where jobs $n - 1$ and $n$ are processed. Here, jobs $n - 1$ and $n$ are in $Lag(s_{t+1})$, and therefore, job $n - 1$ is chosen and migrates to $M_k$, the least-loaded machine. This leads to a state where job $n - 2$ precedes $n$ on $M_i$, a case we have already analyzed.

\begin{figure}[h]
  \centering
  \includegraphics[width=1\textwidth]{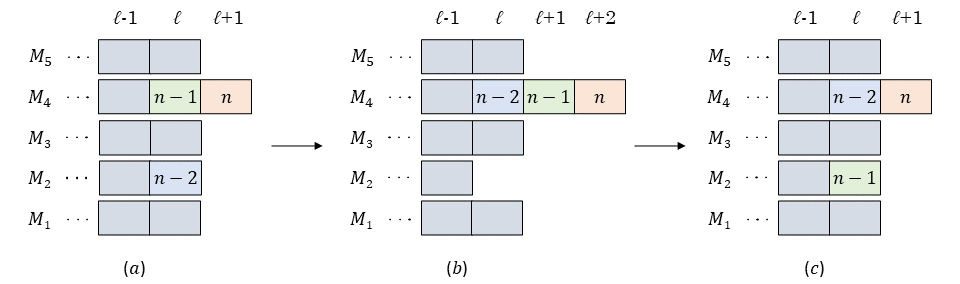}
  \caption{A BRD sequence where the job preceding $n$ on $M_i$ is $n - 1$. (c) The system returns to a schedule where the job preceding $n$ on $M_i$ is $n - 2$.}
  \label{fig_sink_c=1_3}
\end{figure}

In conclusion, for any sink equilibrium $Q \in \mathcal{Q}(G, D)$, we have that $SC(Q)= \ell + 1$. Since the social optimum is $\ceil{\frac{n}{m}} = \ell + 1$ as well, we get $PoSINK(G,D) = 1$.
\end{proof}


\subsection{Related Machines}

Consider the class $\G^{Q2}$. Recall that the two machines are denoted by $M_1$ and $M_2$, such that $r_1 = 1$ and $r_2 = r \leq 1$. The proof of the following theorem adjusts the corresponding proof in~\cite{RST21} for games without competition.

\begin{theorem}
\label{thm:PoAQ2}
PoA$(\G^{Q2}) = $PoS$(\G^{Q2}) = r+1$ if $r \le \frac{\sqrt{5}-1}{2}$, and PoA$(\G^{Q2}) = $PoS$(\G^{Q2}) = \frac{r+2}{r+1}$ if $r > \frac{\sqrt{5}-1}{2}$.
\end{theorem}

\begin{proof}
Let $G \in \G^{Q2}$. Assume $G$ has a NE. Let $W=\sum_i p_i$ be the total processing time of all jobs.
Assume first that  $r \le \frac{\sqrt{5}-1}{2}$. $OPT(G) \ge W/(1+r)$. Also, for any NE $s$, we have that $C_{max}(s) \le W$, since otherwise, the last job to be completed in $s$, which must be processed on the slower machine $M_2$, can migrate to be (at worst) last on the fast machine $M_1$, reducing its completion time. Since its rank in $s$ is $n$, the migration cannot increase its rank. Thus, PoA $\le r+1$.

Assume next that $r>\frac{\sqrt{5}-1}{2}$. Let job $a$ be a last job to complete in a worst NE $s$, $p_1$ be the total processing time of all jobs different from $a$ on $M_1$, and $p_2$ be the total processing time of all jobs different from $a$ on $M_2$ in $s$. If $a$ is processed on $M_1$, then $C_{max}(s) = p_1 + p_a$. Since $s$ is a NE, $p_1 + p_a \leq (p_2 + p_a) / r$, otherwise $a$ could benefit from migration to $M_2$. If $a$ is processed on $M_2$, then $C_{max}(s) = (p_2 + p_a) / r$. Since $s$ is a NE, $(p_2 + p_a) / r \leq p_1 + p_a$, otherwise $a$ could benefit from migration to $M_1$. Therefore, in any case, $C_{max}(s)\leq p_1 + p_a$ and $C_{max}(s)\leq (p_2 + p_a) / r$. Also note that $OPT(G) \ge W/(1+r)$ and $OPT(G) \ge p_a$. Combining these inequalities yields
\[C_{max}(s)\leq \frac{W+p_a}{1+r}\leq \frac{r+2}{r+1}\cdot OPT(G),\]
thus PoA($G$)$\leq \frac{r+2}{r+1}$.

For the PoS lower bound, assume first that $r \le \frac{\sqrt{5}-1}{2}$. Consider an instance consisting of two jobs, $a$ and $b$, where $p_a=1$ and $p_b=1/r$. The priority lists are $\pi_1 = \pi_2 = \langle a,b \rangle$. If $r < \frac{\sqrt{5}-1}{2}$, then the unique NE $s$ is that both jobs are on the fast machine. $C_a(s)=1, C_b(s)=1+1/r$. For every $r < \frac{\sqrt{5}-1}{2}$, it holds that $1+1/r < 1/r^2 $, therefore, job $b$ does not have a beneficial migration. If $r = \frac{\sqrt{5}-1}{2}$, then there is another NE in which job $a$ is on the fast machine, and job $b$ is on the slow machine. Both NE have makespan $1 + 1/r$. An optimal schedule assigns job $a$ on the slow machine and job $b$ on the fast machine, and both jobs complete at time $1/r$. The corresponding PoS is $r+1$.


Assume now that $r > \frac{\sqrt{5}-1}{2}$. 
Consider an instance consisting of three jobs, $x$, $y$ and $z$, where $p_x=1, p_y= r^2+r-1$, and $p_z=r+1$.
The priority lists are $\pi_1 = \langle x, y, z \rangle$ and $\pi_2 = \langle y, x, z \rangle$. 
Note that $p_y \ge 0$ for every $r \ge \frac{\sqrt{5}-1}{2}$. In all NE, job $x$ is on the fast machine, and job $y$ is on the slow machine. Indeed, job $y$ prefers being alone on the slow machine since $r^2 + r > \frac{r^2 + r - 1}{r}$. Job $z$ is indifferent between joining $x$ on the fast machine or $y$ on the slow machine, since $1 + p_z = (p_y + p_z)/ r = r+2$. In an optimal schedule, job $z$ is alone on the fast machine, and jobs $x$ and $y$ are on the slow machine.
Both machines have the same completion time $r+1$.
The PoS is $\frac {r+2}{r+1}$.
\end{proof}

\noindent{\bf The Benefit or Loss due to Competition} The above bounds are identical to those presented in \cite{RST21} for games without competition. Once again, we show that competition may be both helpful and harmful for the social cost. We first describe instances for which the corresponding game without competition has maximal PoS, while the same game with competition has optimal PoA of $1$.

\begin{example}
{\em
Consider the game $G_3$ with $r \le \frac{\sqrt{5}-1}{2}$ and $J = \{a, b, c\}$,\ $p_a = 1 - \epsilon,\ p_b = \epsilon,\ p_c = \frac{1}{r}$. The priority lists are $\pi_1 = \langle a, c, b \rangle,\ \pi_2 = \langle a, b, c \rangle$. In a game without competition, the unique NE is where jobs $a$ and $c$ are processed on the fast machine, and job $b$ is processed on the slow machine (see Figure~\ref{fig:example_3.1}). The makespan is $1 + \frac{1}{r} - \epsilon$. However, with rank-based utilities, this schedule is not stable, since job $a$ can benefit from a migration to $M_1$. Such a migration delays its completion time, but improves its rank from $2$ to $1$, and results in the unique NE in a game with competition. Jobs $a$ and $b$ are processed on the slow machine, while job $c$ is processed alone on the fast machine, yielding an optimal schedule with makespan $\frac{1}{r}$. Thus, for the game $G_3$, the PoS without competition is $\frac{1 + \frac{1}{r} - \epsilon}{\frac{1}{r}} \to r + 1$, while with competition PoA($G_3$) $= 1$.
}

\begin{figure}[h]
  \centering
  \includegraphics[width=1\textwidth]{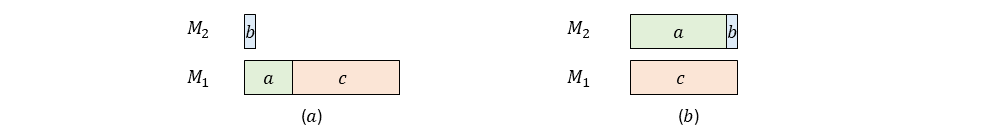}
  \caption{The game $G_3$. (a) NE without competition, (b) NE with competition}
  \label{fig:example_3.1}
\end{figure}

\end{example}

\begin{example}
{\em
Consider the game $G_4$ with $r > \frac{\sqrt{5}-1}{2}$ and $J = \{x, y, z\}$,\ $p_x = 1,\ p_y = r^2 + r - 1,\ p_z = r + 1$. The priority lists are $\pi_1 = \langle x, z, y \rangle,\ \pi_2 = \langle x, y, z \rangle$. Note that $p_y > 0$ for every $r > \frac{\sqrt{5}-1}{2}$. In all NE without competition, job $x$ is on the fast machine, and job $y$ is on the slow machine. Indeed, job $y$ prefers being alone on the slow machine since $r^2 + r > \frac{r^2 + r - 1}{r}$. Job $z$ is indifferent between joining $x$ on the fast machine or $y$ on the slow machine, since $1 + p_z = (p_y + p_z)/ r = r+2$ (see Figure~\ref{fig:example_3.2}). The makespan is $r + 2$. However, with rank-based utilities, this schedule is not stable, since job $x$ can benefit from a migration to $M_1$. Such a migration delays its completion time, but improves its rank from $2$ to $1$, and results in the unique NE in a game with competition. Jobs $x$ and $y$ are processed on the slow machine, while job $z$ is processed alone on the fast machine, yielding an optimal schedule with makespan $r + 1$. Thus, for game $G_4$, the PoS without competition is $\frac{r + 2}{r + 1}$, while with competition PoA($G_4$) $= 1$.
}

\begin{figure}[h]
  \centering
  \includegraphics[width=1\textwidth]{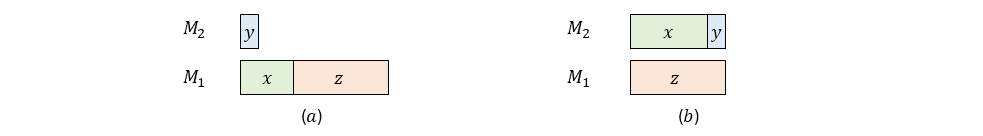}
  \caption{The game $G_4$. (a) NE without competition, (b) NE with competition}
  \label{fig:example_3.2}
\end{figure}

\end{example}

On the other hand, we now show that in $\G^{Q2}$, similarly to $\G^P$, competition might also increase the inefficiency. Specifically, in the following example, we describe an instance $G_5$ such that without competition $PoA(G_5)=1$, and with competition $PoS(G_5) = 1 + \frac{r}{r+1}$.

\begin{example}
{\em
Consider the game $G_5$ with $4$ jobs, $J = \{a, b, c, d\}$, where $p_a = r^2,\ p_b = r + 1 - r^2,\ p_c = \epsilon,\ p_d = r^2 + r - \epsilon$. The priority lists are $\pi_1 = \langle a, b, c, d \rangle$ and $\pi_2 = \langle a, c, d, b \rangle$. The no-competition unique NE is an optimal schedule in which jobs $a$ and $b$ are processed on $M_1$, and jobs $c$ and $d$ are processed on $M_2$ (see Figure~\ref{fig:example_4.1}). The makespan is $r + 1$. However, with rank-based utilities, this schedule is not stable, since job $a$ can benefit from migration to $M_2$. Such a migration delays its completion time, and improves its rank from $2$ to $1$. In the unique NE with competition, jobs $a$ and $c$ are processed on $M_2$, while jobs $b$ and $d$ are processed on $M_1$. Indeed, job $d$ prefers $M_1$ since $p_b + p_d < (p_a + p_c + p_d)/ r$. The makespan is $2r + 1 - \epsilon$. Thus, for the game $G_5$, the PoA with competition is $1$, while without competition PoS($G_5$) $= \frac{2r + 1 -\epsilon}{r+1} \to 1 + \frac{r}{r+1}$.
}

\begin{figure}[h]
  \centering
  \includegraphics[width=1\textwidth]{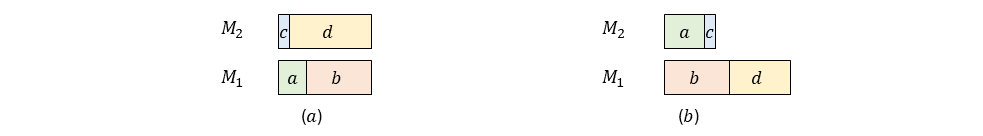}
  \caption{The game $G_5$. (a) NE without competition, (b) NE with competition}
  \label{fig:example_4.1}
\end{figure}

\end{example}

\subsubsection{Sink Equilibria Analysis - Related Machines}

We now analyze the price of sinking for games in $\G^{Q2}$ with arbitrary job lengths. As we show, the price of sinking - regardless of the deviator rule applied - cannot be bounded by any constant, since it depends on $1/r$.

\begin{example}
{\em
Consider the game $G'$ with two jobs $a$ and $b$, where $p_a = 1$, $p_b = r$, and $\pi_1 = \pi_2 = \langle a, b \rangle$. $G'$ has no pure NE. Profile $s_1$ is the social optimum. The sink consists of $4$ profiles  (see Figure~\ref{fig:example_5}), and has price of sinking of $\frac {r + 3}{4} + \frac{1}{2r}$. The high cost of the sink is due to the fact that the long job is on the slow machine in two out of the four profiles.
}

\begin{figure}[h]
  \centering
  \includegraphics[width=1\textwidth]{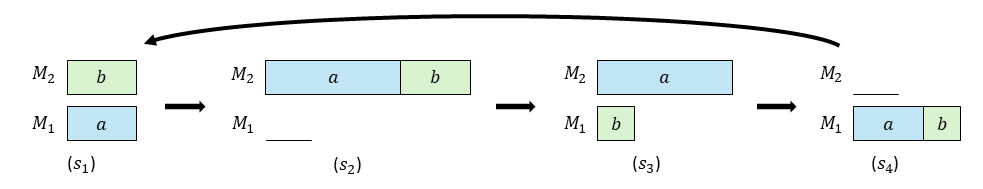}
  \caption{The sink equilibrium of $G'$. Job $a$ benefits from a migration in $s_1$ and $s_3$, job $b$ benefits from a migration in $s_2$ and $s_4$.}
  \label{fig:example_5}
\end{figure}
\end{example}

\section{Extension: Games with Competition Classes}
\label{sec:competition_classes}

In this section we consider a natural extension of our setting, in which the structure of the competition is more involved. Formally, let $S$ be is a partition of the jobs into competition sets, i.e., $S = \{S_1, \ldots, S_c\}$ such that $c \leq n$, $\bigcup_{l=1}^{c} S_l = J$, and for all $l_1 \neq l_2$, we have $S_{l_1} \cap S_{l_2} = \emptyset$. For every job $j \in S_l$, the competitors of $j$ are the other jobs in $S_l$. Let $n_l$ denote the number of jobs in $S_l$.

The goal of a player is to do well relative to its competitors. Formally, for a profile $s$, let $C^{s,l} = \langle C_{1}^{s,l}, \ldots, C_{n_l}^{s,l} \rangle$ be a sorted vector of the completion times of the players in $S_l$. That is, $C_{1}^{s,l} \leq \ldots \leq C_{n_l}^{s,l}$, where $C_{1}^{s,l}$ is the minimal completion time of a player from $S_l$ in $s$, etc. The $rank$ of player $j \in S_l$ in profile $s$, denoted by $rank_j(s)$ is the rank of its completion time in $C^{s,l}$.

So far, throughout this work we assumed $c=1$; that is, $S_1=J$. Classical job scheduling games, in which the objective of a job is to minimize its cost, can be viewed as a special case of $c=n$; that is $S_i=\{i\}$. In this case, the rank of each job is $1$, independent of the schedule, and the secondary objective, of minimizing the completion time, becomes the only objective of a job. 

The following example demonstrates how the competition structure is crucial for the stability of the game. Consider a simple game $G$ with $n = 3$ unit-jobs $a,\ b,\ c$, assigned to $m = 2$ identical machines, with a global priority list $\pi = \langle a, b, c \rangle$. 
If $S_1 = J$, then any profile where jobs $a$ and $c$ are assigned to one machine while job $b$ is assigned to the other is stable (see Figure~\ref{fig:competition_sets}(a)).  However, if $S_1 = \{a, b\}$ and $S_2 = \{c\}$, then $G$ has no NE, as job $a$ aims to be first, and job $b$ is unwilling to settle for second place. A sink equilibrium is presented in Figure~\ref{fig:competition_sets}(b). Finally, if $S_1 = \{a\}, S_2 = \{b\}, S_3 = \{c\}$ then any schedule where all jobs are not assigned to the same machine is a NE (Figure~\ref{fig:competition_sets}(c)).

\begin{figure}[h]
  \centering
  \includegraphics[width=1\textwidth]{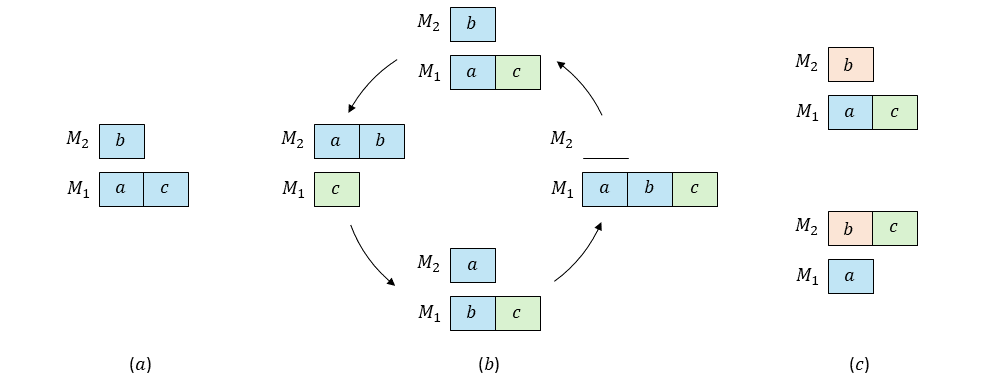}
  \caption{Schedules of $\{a,b,c\}$ with various competition structure, and $\pi = \langle a, b, c \rangle$. (a) The unique (up to machine's renaming) NE for $S_1 = J$, (b) A sink-equilibrium for $S_1 = \{a, b\}$ and $S_2 = \{c\}$, (c) NE schedules for $S_1 = \{a\},\ S_2 = \{b\},\ S_3 = \{c\}$.}
  \label{fig:competition_sets}
\end{figure}

The example clarifies that our positive results, even for unit-length jobs, and two identical machines with a global priority list, do not carry over to games with arbitrary competition. 
On the other hand, for games with $m = 2$ identical machines, arbitrary job lengths and {\em Inveresed Policies}, Algorithm~\ref{alg:AlgNext} produces a NE for any partition of $J$ into competition sets. It is easy to verify that the proof of Theorem~\ref{Thm:InvPol_Convergence} is valid for any competition structure. Thus, we have,

\begin{theorem}
Algorithm~\ref{alg:AlgNext} on two machines with inversed policies produces a NE for any partition of $J$ into competition sets.
\end{theorem}

The following section provides additional positive results for arbitrary competition -- assuming a 'dynamic' priority list, based on seniority.

\subsection{Games with Competition-Sets and Seniority-Based Priorities}

In real life applications, users tend to compete with similar users. This phenomenon motivates the study of games in which the competition sets are defined based on certain job parameters. A natural assumption for this setting is that the priority list of a machine is not given by a permutation of the jobs in $J$, but by a permutation of the $c$ competition sets. For instance, competition sets could be defined by job lengths, meaning that all jobs of a specific length form a competition set, with machines applying rules such as SPT (shortest processing time) or LPT (longest processing time) to determine the jobs' processing order. The internal order within each set is determined by a {\em seniority-based} tie-breaking method: in the initial assignment, ties are broken arbitrarily. However, once the initial assignment is made, if a job deviates to another machine, it is placed last among the jobs in its competition set on that machine, as those jobs have a seniority advantage. 

We show that this model obeys the following property:

\begin{claim}
\label{cl:senior}
In a game with competition sets and seniority-based internal order, a deviation is beneficial if and only if it reduces the completion time of the deviating job.
\end{claim}

\begin{proof}
Consider a deviation from profile $s$ to profile $s'$ in which the completion time of the deviating job is not reduced. Specifically, let job $a \in S_k$ deviate from machine $M_1$ to machine $M_2$, such that $C_a(s') \geq C_a(s)$. We show that $rank_a(s') \ge rank_a(s)$.

Let $b \in S_k$, and assume $C_b(s) < C_a(s)$. If $b$ is assigned to $M_2$, then by the seniority rule, $C_b(s') < C_a(s')$. If $b$ is not on $M_2$, then since $C_a(s') \ge C_a(s)$ and $C_b(s') = C_b(s)$, we have that $C_b(s') < C_a(s')$. Also, if $C_b(s) = C_a(s)$, then since $C_a(s') \ge C_a(s)$ and $C_b(s') = C_b(s)$, we have that $C_b(s') \le C_a(s')$. Therefore, the rank of $b$ could not reduce. 
\end{proof}

The above claim implies that the corresponding rank-based utility game is equivalent to a competition-free game with a traditional cost-minimization objective.
These games were analyzed in~\cite{RST21}, where it is shown that a NE exists for any instance on identical machines. The next theorem follows from combining Claim~\ref{cl:senior} with the analysis in~\cite{RST21}.

\begin{theorem}
\label{Thm:classes_seniority}
Let $G$ be a game with competition sets and seniority-based internal order, then $G$ has a NE, and BRD converges to a NE.
\end{theorem}


\section{Conclusions and Directions for Future Work}

In this work, we explored the impact of competition on coordination mechanisms, where players aim to optimize their rank rather than their completion time. Our study centered on the existence and computation of pure NE. Our findings reveal that in general, introducing competition can make the system less stable. While competition preserves certain system-wide properties, such as the NP-completeness of determining equilibrium in the general case, and the PoA bounds seen in traditional scheduling games, it introduces new challenges. On the positive side, we characterized classes of games where equilibria are guaranteed to exist, demonstrated that for some games the presence of competition improves the equilibria inefficiency, and proposed a specific mechanism for  identical machines that stabilize the system.

Our work leaves several open questions for future work. A primary question concerns the existence and computation of NE in games with unit-length jobs and machine-dependent priority lists, when played on an arbitrary number of machines, or even a constant number greater than two. Extending our results to these broader settings could provide deeper insights into the stability and efficiency of competitive scheduling games. Additionally, designing scheduling policies that lead to stable outcomes while ensuring good social value, for games with more than two machines, remains an intriguing challenge.

While it is straightforward to observe that in $\G^{global}$ every NE is strong, a natural extension of our work would involve investigating the existence and computation of strong Nash equilibria in other game classes.

Another potential direction is to explore alternative deviator rules for cases where a pure NE does not exist.  Moreover, the study of the price of sinking with specific deviator rules is relevant to many other games in which a NE is not guaranteed to exist.

Finally, the study of rank-based utilities is relevant to many other environments managed by multiple strategic users. As demonstrated in this work, the shift from a traditional cost-minimization objective to a rank-minimization objective creates an entirely different game, presenting new challenges both in computing stable outcomes and analyzing their quality.




\newpage
\begin{thebibliography}{10}

\bibitem{AD+08}
E.~Anshelevich, A.~Dasgupta, J.~Kleinberg, E.~Tardos, T.~Wexler, and
  T.~Roughgarden.
\newblock The price of stability for network design with fair cost allocation.
\newblock {\em SIAM Journal on Computing}, 38(4):1602--1623, 2008.

\bibitem{AKT08}
I.~Ashlagi, P.~Krysta, and M.~Tennenholtz.
\newblock Social context games.
\newblock In {\em International Workshop on Internet and Network Economics},
  pages 675--683. Springer, 2008.

\bibitem{AART06}
B.~Awerbuch, Y.~Azar, Y.~Richter, and D.~Tsur.
\newblock Tradeoffs in worst-case equilibria.
\newblock {\em Theoretical Computer Science}, 361(2):200--209, 2006.

\bibitem{AJM08}
Y.~Azar, K.~Jain, and V.~Mirrokni.
\newblock (almost) optimal coordination mechanisms for unrelated machine
  scheduling.
\newblock In {\em Proceedings of the $19$th Annual ACM-SIAM Symposium on
  Discrete Algorithms}, SODA '08, pages 323--332, 2008.

\bibitem{BFNR11}
N.~Berger, M.~Feldman, O.~Neiman, and M.~Rosenthal.
\newblock Dynamic inefficiency: Anarchy without stability.
\newblock In {\em Algorithmic Game Theory: 4th International Symposium, SAGT
  2011, Amalfi, Italy, October 17-19, 2011. Proceedings 4}, pages 57--68.
  Springer, 2011.

\bibitem{BFHS09}
F.~Brandt, F.~Fischer, P.~Harrenstein, and Y.~Shoham.
\newblock Ranking games.
\newblock {\em Artificial Intelligence}, 173(2):221--239, 2009.

\bibitem{CF19}
I.~Caragiannis and A.~Fanelli.
\newblock An almost ideal coordination mechanism for unrelated machine
  scheduling.
\newblock {\em Theory of Computing Systems}, 63(1):114--127, 01 2019.

\bibitem{CKN04}
G.~Christodoulou, E.~Koutsoupias, and A.~Nanavati.
\newblock Coordination mechanisms.
\newblock In {\em Automata, Languages and Programming: 31st International
  Colloquium, {ICALP} 2004, Turku, Finland, July 12-16, 2004. Proceedings},
  pages 345--357, 2004.

\bibitem{CDN11}
J.~Cohen, C.~Dürr, and T.~N. Kim.
\newblock Non-clairvoyant scheduling games.
\newblock {\em Theory of Computing Systems}, 49(1):3--23, 2011.

\bibitem{CQ12}
J.~R. Correa and M.~Queyranne.
\newblock Efficiency of equilibria in restricted uniform machine scheduling
  with total weighted completion time as social cost.
\newblock {\em Naval Research Logistics (NRL)}, 59(5):384--395, 2012.

\bibitem{CzumajV07}
A.~Czumaj and B.~V\"{o}cking.
\newblock Tight bounds for worst-case equilibria.
\newblock {\em ACM Trans. Algorithms}, 3(1):4:1--4:17, 2007.

\bibitem{EKM03}
E.Even-Dar, A.Kesselman, and Y.Mansour.
\newblock Convergence time to nash equilibria.
\newblock In {\em Proc.\ 30th Int. Colloq. on Automata, Languages, and
  Programming}, pages 502--513, 2003.

\bibitem{FST17}
M.~Feldman, Y.~Snappir, and T.~Tamir.
\newblock The efficiency of best-response dynamics.
\newblock In {\em The 10th International Symposium on Algorithmic Game Theory
  (SAGT)}, 2017.

\bibitem{FT15}
M.~Feldman and T.~Tamir.
\newblock Convergence of best-response dynamics in games with conflicting
  congestion effects.
\newblock {\em Information Processing Letters}, 115(2):112--118, 2015.

\bibitem{GLM10}
M.~Gairing, T.~Lücking, and M.~M. et~al.
\newblock Computing nash equilibria for scheduling on restricted parallel
  links.
\newblock {\em Theory of Computing Systems}, 47(2):405--432, 2010.

\bibitem{GMV05}
M.~X. Goemans, V.~S. Mirrokni, and A.~Vetta.
\newblock Sink equilibria and convergence.
\newblock In {\em 46th Annual {IEEE} Symposium on Foundations of Computer
  Science {(FOCS} 2005), 23-25 October 2005, Pittsburgh, PA, USA, Proceedings},
  pages 142--154, 2005.

\bibitem{GGKV13}
L.~A. Goldberg, P.~W. Goldberg, P.~Krysta, and C.~Ventre.
\newblock Ranking games that have competitiveness-based strategies.
\newblock {\em Theoretical Computer Science}, 476:24--37, 2013.

\bibitem{ILMS09}
N.~Immorlica, L.~E. Li, V.~S. Mirrokni, and A.~S. Schulz.
\newblock Coordination mechanisms for selfish scheduling.
\newblock {\em Theor. Comput. Sci.}, 410(17):1589--1598, 2009.

\bibitem{Kann91}
V.~Kann.
\newblock Maximum bounded $3$-dimensional matching is max snp-complete.
\newblock {\em Information Processing Letters}, 37:27--35, 1991.

\bibitem{KBL13}
B.~Kawald and P.~Lenzner.
\newblock On dynamics in selfish network creation.
\newblock In {\em Proceedings of the twenty-fifth annual ACM symposium on
  Parallelism in algorithms and architectures}, pages 83--92, 2013.

\bibitem{K13}
K.~Kollias.
\newblock Nonpreemptive coordination mechanisms for identical machines.
\newblock {\em Theory of Computing Systems}, 53(3):424--440, 2013.

\bibitem{KP09}
E.~Koutsoupias and C.~Papadimitriou.
\newblock Worst-case equilibria.
\newblock {\em Computer Science Review}, 3(2):65--69, 2009.

\bibitem{LST90}
J.~K. Lenstra, D.~B. Shmoys, and E.~Tardos.
\newblock Approximation algorithms for scheduling unrelated parallel machines.
\newblock {\em Math. Program.}, 46(3):259–271, 1990.

\bibitem{LY12}
P.~Lu and C.~Yu.
\newblock Worst-case nash equilibria in restricted routing.
\newblock {\em Journal of Computer Science and Technology}, 27(4):710--717,
  2012.

\bibitem{RST21}
V.~{Ravindran Vijayalakshmi}, M.~Schr{\"o}der, and T.~Tamir.
\newblock Scheduling games with machine-dependent priority lists.
\newblock {\em Theoretical Computer Science}, 855:90--103, 2021.

\bibitem{RT23}
S.~Rosner and T.~Tamir.
\newblock Scheduling games with rank-based utilities.
\newblock {\em Games and Economic Behavior}, 140:229--252, 2023.

\end{thebibliography}
\end{document}